\documentclass{article}
\usepackage{amsmath,natbib}
\usepackage{amssymb}
\usepackage{amsthm,bbm,mathtools}

\newcommand{\one}{\mathbbm 1}
\newcommand{\ap}[1]{\prescript{a}{}{#1}}

\def\reals{\mathbb{R}}

\def\ereals{\overline{\mathbb{R}}}

\def\comp{\raise 1pt \hbox{$\scriptstyle\circ$}}

\def\minimize{\mathop{\rm minimize}\limits}

\def\st{\mathop{\rm subject\ to}}

\def\dom{\mathop{\rm dom}\nolimits}

\def\upto{{\raise 1pt \hbox{$\scriptstyle \,\nearrow\,$}}}
\def\downto{{\raise 1pt \hbox{$\scriptstyle \,\searrow\,$}}}
\def\inte{\mathop{\rm int}}

\def\cl{\mathop{\rm cl}\nolimits}

\def\epi{\mathop{\rm epi}\nolimits}

\def\FF{(\F_t)_{t=0}^T}
\def\one{\mathbbm 1}
\def\ovr{\mathop{\rm over}}

\def\A{{\cal A}}
\def\B{{\cal B}}
\def\C{{\cal C}}

\def\D{{\cal D}}

\def\F{{\cal F}}

\def\L{{\cal L}}
\def\M{{\cal M}}
\def\N{{\cal N}}

\def\P{{\cal P}}

\def\Q{{\cal Q}}

\def\U{{\cal U}}

\def\Y{{\cal Y}}

\newtheorem{theorem}{Theorem}
\newtheorem{lemma}[theorem]{Lemma}
\newtheorem{corollary}[theorem]{Corollary}

\newtheorem{example}{Example}
\newtheorem{remark}{Remark}
\theoremstyle{definition}

\newtheorem{assumption}{Assumption}

\title{Convex duality in optimal investment and contingent claim valuation in illiquid markets}
\author{Teemu Pennanen\thanks{Department of Mathematics, King's College London,
Strand, London, WC2R 2LS, United Kingdom, {\it teemu.pennanen@kcl.ac.uk}.} \and Ari-Pekka Perkki\"o\thanks{Department of Mathematics, Technische Universit\"at Berlin, Building MA,
Str. des 17. Juni 136, 10623 Berlin, Germany, {\it perkkioe@math.tu-berlin.de}. The author is grateful to the Einstein Foundation for the financial support.}}

\begin{document}

\maketitle

\begin{abstract}
This paper studies convex duality in optimal investment and contingent claim valuation in markets where traded assets may be subject to nonlinear trading costs and portfolio constraints. 
Under fairly general conditions, the dual expressions decompose into tree terms, corresponding to the agent's risk preferences, trading costs and portfolio constraints, respectively. The dual representations are shown to be valid when the market model satisfies an appropriate generalization of the no-arbitrage condition and the agent's utility function satisfies an appropriate generalization of asymptotic elasticity conditions. When applied to classical liquid market models or models with bid-ask spreads, we recover well-known pricing formulas in terms of martingale measures and consistent price systems. Building on the general theory of convex stochastic optimization, we also derive optimality conditions in terms of an extended notion of a ``shadow price''.
\end{abstract}

\section{Introduction}

A fundamental problem in financial economics is that of managing investments so that their proceeds cover given liabilities as well as possible. This is important not only in wealth management but also in valuation of financial products. Both optimal investment and asset valuation are often analyzed by introducing certain dual variables that appear e.g.\ in optimality conditions and valuation operators. 
Classical references include \cite{hk79}, \cite{hp81}, \cite{kre81} and \cite{dmw90} where the no-arbitrage condition was related to the existence of martingale measures; see \cite{ds6} for a detailed discussion of the topic. In problems of optimal investment, convex duality has become an important tool in the analysis of optimal solutions; see e.g.\ \cite{ck92}, \cite{ks99}, \cite{dgrsss}, \cite{kz3},~\cite{rog13} and their references.


This article studies convex duality in optimal investment and contingent claim valuation in illiquid markets where one may face nonlinear trading costs and portfolio constraints that hinder transfer of funds between asset classes or through time. Building on general theory of convex stochastic optimization, we give explicit expressions for the dual objective as well as sufficient conditions for absence of a duality gap and for the attainment of the dual optimum. The first set of conditions involve a multivariate generalizations of the asymptotic elasticity conditions from \cite{ks99,sch1} as well as a condition that generalizes the robust no-arbitrage condition in unconstrained models with proportional transaction costs. The conditions for attainment of the dual optimum are closely related to those of \cite[Corollary~5.2]{bc11} in the context of perfectly liquid market models. The dual solutions provide a natural extension of the notion of a ``shadow price'' from unconstrained sublinear market model of \cite{cms14} to general convex models. Analogously to shadow prices, the optimal dual variables, when they exist, yield pointwise optimality conditions for optimal investment strategies.

Under nonlinear illiquidity effects and nonconical portfolio constraints, the dual objective turns out to be the sum of three nonlinear terms. Like in classical perfectly liquid models, the first term is the conjugate of the primal objective. The two other terms correspond to trading costs and portfolio constraints, respectively. In the classical perfectly liquid model with a cash-account, we recover the familiar dual problem over positive multiples of martingale densities. In the absence of a cash-account, however, the dual variables become stochastic discount factors that account for uncertainty as well as for time-value of money much like in the numeraire-free conical models of \cite{dr91}, \cite{jk1}, \cite{nap1}, \cite{mad15} and \cite{teh??}. Introduction of short-sales constraints result in stochastic discount factors that turn the market prices into supermartingales like in \cite{jk95b}. Adding proportional transaction costs, one gets similar conditions for a ``shadow price'' process that evolves within the bid-ask spread; see \cite{jk95a,cms14}. 

Our model of nonlinear trading costs is closely related to the continuous-time model of \cite{gr15}. Within the discrete-time framework, we are able to relax the coercivity conditions of \cite{gr15} so as to allow for nondecreasing as well as linear trading costs. Analogously to \cite{gr15}, we extend the analysis also to contingent claims with physical delivery. It turns out that this is the natural framework of analysis and the results on cash-valued contingent claims are derived from the corresponding results for multivariate claims. 


Building on the duality theory of optimal investment, we derive dual expressions for the values of contingent claims that may have several payout dates. As in \cite{pen14}, two widely studied notions of a value will be considered. The first one, which we will refer to as {\em accounting value}, is defined as the least amount of initial cash that is needed to hedge a claim at an acceptable level of risk. In the most risk averse case, where one does not accept any risk of a loss at maturity, this reduces to the classical superhedging cost. More reasonable versions are obtained by relaxing the superhedging requirement to that of acceptable ``risk'' associated with the terminal net position. Such valuations have been studied under various names including (but not restricted to) ``risk measure'' \cite{adeh99,adk9}, ``efficient hedging'' \cite{fl,fs11} or ``good-deal bounds'' \cite{ch2}. The second notion of value is that of {\em indifference price} which is the natural notion for an agent who considers entering a financial transaction and is interested in how the trade would change his existing financial position; see e.g. \cite{hn89,re0,car9} and their references. 

Dual expressions for the accounting values turn out to have the same general structure as in ``good deal bounds'' studied e.g.\ in \cite{jk1} and \cite{ch2}. The dual representation may also be seen as a multiperiod generalization of the dual representation of a convex risk measure in illiquid markets. In the case of indifference swap rates, we obtain illiquid discrete-time versions of the pricing formulas derived e.g.\ in~\cite[Section~7]{oz9} and~\cite[Section~4]{bfg11}.

\section{Optimal investment and contingent claim valuation}\label{sec:oicv}

Consider a financial market where a finite set $J$ of assets can be traded over $t=0,\ldots,T$. As usual, we will model uncertainty and information by a filtered probability $(\Omega,\F,\FF,P)$. The cost of buying a portfolio $x\in\reals^J$ at time $t$ and state $\omega$ will be denoted by $S_t(x,\omega)$. We assume that $S_t:\reals^J\times\Omega\to\ereals$ is a an $\F_t$-measurable convex normal integrand (see the appendix) such that $S_t(0,\omega)=0$. This abstract specification covers e.g.\ proportional transaction costs as well as transient illiquidity effects one faces e.g.\ in limit order markets. In the presence of portfolio constraints, the set of portfolios that can be held over $(t,t+1]$ will denoted by $D_t(\omega)\subseteq\reals^J$. We assume that $D_t$ is $\F_t$-measurable, closed and convex-valued with $0\in D_t(\omega)$. We emphasize that we do not assume a priori the existence of a perfectly liquid numeraire asset that is free of such constraints. Such models can, however, be treated as special cases of the market model $(S,D)$; see Example~\ref{ex:liquid} below. More examples with e.g.\ proportional transactions and specific instances of portfolio constraints can be found in \cite{pen11,pen11b,pen12,pen14}.

In a market without perfectly liquid assets it is important to distinguish between payments at different points in time. See \cite{pen14,mad15,teh??} for further discussion. Accordingly, we will describe the agent's preferences over sequences of payments by a normal integrand $V$ on $\reals^{T+1}\times\Omega$. More precisely, the agent prefers to make a random sequence $c^1$ of payments over another $c^2$ if
\[
EV(c^1)<EV(c^2)
\]
while she is indifferent between the two if the expectations are equal. A possible choice would be
\begin{equation}\label{separable}
V(c,\omega)=\sum_{t=0}^TV_t(c_t,\omega).
\end{equation}
In general, the function $V$ can be seen as a multivariate generalization of a ``loss function'' in the sense of \cite[Definition~4.111]{fs11}, where a loss function was a nonconstant nondecreasing function on $\reals$. Multivariate utility functions have been studied also in \cite{co11} where the utility was a function of a portfolio of assets. More recently, \cite{teh??} studied optimal consumption investment

We allow for random, nondifferentiable, extended real-valued loss functions $V$ but will require the following.

\begin{assumption}\label{ass1}
The loss function $V$ is a normal integrand on $\reals^{T+1}\times\Omega$ such that for $P$-almost every $\omega\in\Omega$, the function $V(\cdot,\omega)$ is convex and nondecreasing with respect to the componentwise order of $\reals^J$, $V(0,\omega)=0$ and 
for every nonzero $c\in\reals^{T+1}_+$ there exists an $\alpha>0$ such that $V(\alpha c,\omega)>0$.
\end{assumption}

The last condition in Assumption~\ref{ass1} means that the agent is not completely indifferent with respect to nonzero nonnegative payments. It holds in particular if $V$ is of the form \eqref{separable} and each $V_t$ is convex $\F_t$-normal integrand on $\reals\times\Omega$ such that $V_t(\cdot,\omega)$ is nonconstant and nondecreasing with $V(0,\omega)=0$. 

We will study optimal investment from the point of view of an agent who has financial liabilities described by a random sequence $c$ of payments to be made over time. We allow $c$ to take arbitrary real values so it can describe endowments as well as liabilities. The agent's problem is to find a trading strategy that hedges against the liabilities $c$ as well as possible as measured by $EV$. Denoting the linear space of adapted $\reals^J$-valued trading strategies by $\N$, we can write the problem as
\begin{equation}\label{alm}\tag{ALM}
\minimize\quad E V(S(\Delta x)+c)\quad\ovr\quad x\in\N_D,
\end{equation}
where $\N_D=\{x\in\N\mid x_t\in D_t\ \forall t\ P\text{-a.s.}\}$ is the set of feasible trading strategies, $S(\Delta x)$ denotes the process $(S_t(\Delta x_t))_{t=0}^T$ and $x_{-1}:=0$. Here and in what follows, we define the expectation of a measurable function as $+\infty$ unless the positive part is integrable. For any $z\in\reals^J$ and $c\in\reals$, we set $V(S(z,\omega)+c,\omega):=+\infty$ unless $z_t\in\dom S_t(\cdot,\omega)$ for all $t$. Thus, \eqref{alm} incorporates the constraint $\Delta x_t\in\dom S_t$ almost surely for all $t$, i.e.\ there may be limits on traded amounts. Assumption~\ref{ass1} guarantees that the objective of \eqref{alm} is a well-defined convex function on $\N$; see Corollary~\ref{cor:vs} in the appendix. We assume that $D_T(\omega)=\{0\}$, i.e., the agent is required to close all positions at time $T$.

Throughout this paper, we will denote the optimum value of \eqref{alm} by $\varphi(c)$. The {\em optimum value function} $\varphi$ is an extended real-valued function on the linear space $\M$ of random sequences of payments.

Problem \eqref{alm} covers discrete-time versions of many more traditional models of optimal investment that have appeared in the literature. In particular, \eqref{alm} can be interpreted as an optimal consumption-investment problem with random endowment $-c_t$ and consumption $-S_t(\Delta x_t)-c_t$ at time $t$. Our formulation extends the more common formulations with a perfectly liquid numeraire asset; see e.g.\ \cite{ck92,kz3} and their references. Problem~\ref{alm} is close to \cite[Section~2]{teh??}, where optimal investment-consumption was studied in a linear discrete-time market model without a perfectly liquid numeraire asset. 

Allowing for extended real-valued loss functions $V$, we can treat superhedging problems and problems with explicit budget constraints as special cases of \eqref{alm}. We will denote by
\[
\C := \{c\in\M\,|\,\exists x\in\N_D:\ S(\Delta x) + c\le 0\},
\]
the {\em set of claims that can be superhedged without a cost}.

\begin{example}\label{ex:terminal}
When $V$ is of the form
\[
V(c,\omega)=
\begin{cases}
  V_T(c_T,\omega) & \text{if $c_t\le 0$ for $t<T$},\\
  +\infty & \text{otherwise},
\end{cases}
\]
problem \eqref{alm} can be written with explicit budget constraints as follows
\begin{equation}\label{alm2}
\begin{aligned}
&\minimize\quad & &EV_T(S_T(\Delta x_T)+c_T) \quad\ovr \quad x\in\N_D\\
&\st & &S_t(\Delta x_t)+c_t\le 0,\quad \forall t<T\quad P\text{-a.s}.
\end{aligned}
\end{equation}
This is an illiquid version of the classical problem of maximizing expected utility of terminal wealth. Indeed, since $x_T=0$ for all $x\in\N_D$, we have $S_T(\Delta x_T)=S_T(-x_{T-1})$, which is the cost of closing all positions at time $T$. Alternatively, one may interpret $-S_T(-x_{T-1})$ as the liquidation value of $x_{T-1}$. When\footnote{Here and in what follows, $\delta_C$ denotes the {\em indicator function} of a set $C$: $\delta_C(x)$ equals $0$ or $+\infty$ depending on whether $x\in C$ or not.} $V=\delta_{\reals^{T+1}}$ we have $\varphi=\delta_\C$.
\end{example}

In models with perfectly liquid cash, problem \eqref{alm2} can be written in a more familiar form.

\begin{example}\label{ex:liquid}
Models with perfectly liquid cash correspond to
\[
S_t(x,\omega) = x^0 + \tilde S_t(\tilde x,\omega)\quad\text{and}\quad D_t(\omega) = \reals\times\tilde D_t(\omega),
\]
where $x=(x^0,\tilde x)$ and $\tilde S$ and $\tilde D$ are the trading cost and the constraints for the remaining risky assets $J\setminus\{0\}$. When $c$ is adapted, one can then use the budget constraint in problem \eqref{alm2} to substitute out the cash investments $x^0$ from the problem formulation 
much like in \cite[Equation~(2.9)]{cr7} and \cite[Equation~(4)]{gr15}. Recalling that $x_T=0$ for $x\in\N_D$, problem \eqref{alm2} can thus be written as
\[
\begin{aligned}
&\minimize\quad & &EV_T\left(\sum_{t=0}^T\tilde S_t(\Delta\tilde x_t)+\sum_{t=0}^Tc_t\right)\quad\ovr\quad x\in\N_D
\end{aligned}
\]
If one specializes further to the perfectly liquid market model where $\tilde S_t(\tilde x)=\tilde s_t\cdot\tilde x_t$ for a unit price process $\tilde s$ that is independent of the trades, one can express the accumulated trading costs in the objective as a stochastic integral of $\tilde x$ with respect to $\tilde s$, thus recovering a discrete-time version of the classical formulation from e.g.\ in \cite{ks99,dgrsss,ks3,bf8,bfg11}. Note that in~\cite{ks99,ks3,bf8}, the financial position of the agent was described solely in terms of an initial endowment $w\in\reals$ without future liabilities. This corresponds to $c_0=-w$ and $c_t=0$ for $t>0$.
\end{example}

We finish this section with a short review of contingent claim valuation under illiquidity; see \cite{pen14} for further discussion. The {\em accounting value} of the liability of delivering a contingent claim $c\in\M$ is defined in terms of the optimum value function $\varphi$ of \eqref{alm} as
\[
\pi^0_s(c) = \inf\{\alpha\in\reals\,|\,\varphi(c-\alpha p^0)\le 0\},
\]
where $p^0\in\M$ is given by $p^0_0=1$ and $p^0_t=0$ for $t>0$. The accounting value gives the least amount of initial capital needed to hedge the claim $c$ at an acceptable level of risk as measured by $EV$. This extends the notion of ``efficient hedging'' with minimal shortfall risk from \cite[Section~8.2]{fs11} to illiquid markets and claims with multiple payout dates. Alternatively, $\pi^0_s(c)$ can be seen as an extension of a ``risk measure'' from \cite{adeh99} to markets without a perfectly liquid ``reference asset''. The function $\pi^0_s:\M\to\ereals$ is convex and nondecreasing with respect to the pointwise ordering on $\M$. It also has a translation property that extends the one proposed in \cite{adeh99}; see the remarks after \cite[Theorem~3.2]{pen14}. Applications of $\pi^0_s$ in pension insurance are studied in \cite{hkp11b}. While $\pi^0_s(c)$ gives the least amount of initial cash required to construct an acceptable hedging strategy for a liability $c\in\M$, the number
\[
\pi^0_l(c):=\sup\{\alpha\in\reals\,|\,\varphi(\alpha p^0-c)\le 0\}
\]
gives the greatest initial {\em payment} one could cover at an acceptable level of risk (by shorting traded assets at time $t=0$ and dynamically trading to close positions by time $t=T$) when receiving $c\in\M$. We will call $\pi^0_l(c)$ the {\em accounting value of an asset} $c\in\M$. Clearly, $\pi^0_l(c)=-\pi^0_s(-c)$. Accounting values for assets are relevant e.g.\ in banks whose assets are divided into the {\em trading book} and the {\em banking book}. Assets in the banking book are typically illiquid assets without secondary markets. The above definition of $\pi^0_l(c)$ gives a market consistent, hedging based value for such assets.

The second notion of a ``value'' considered in this paper extends the notion of indifference price from \cite{hn89}. The {\em indifference swap rate} of exchanging a claim $c\in\M$ for a multiple of a premium sequence $p\in\M$ is defined as
\[
\pi_s(\bar c,p;c) = \inf\{\alpha\in\reals\,|\, \varphi(\bar c+c-\alpha p)\le\varphi(\bar c)\},
\]
where $\bar c\in\M$ describes the agent's initial liabilities. The interpretation is that the indifference swap rate gives the least swap rate that would allow the agent to enter the swap contract without worsening her position as measured by the optimum value of \eqref{alm}. While accounting values are important in accounting and financial supervision, indifference swap rates are more relevant in trading. The indifference swap rate for taking the other side (``long position'') of the swap contract is given analogously by
\[
\pi_l(\bar c,p;c) = \sup\{\alpha\in\reals\,|\, \varphi(\bar c-c+\alpha p)\le\varphi(\bar c)\}.
\]
Note that when $c=(0,\ldots,0,c_T)$ and $p=p^0$, the indifference swap rate reduces to the familiar notion of indifference price.

We end this section by recalling some basic facts about superhedging and the arbitrage bounds for accounting values and indifference swap rates; see \cite{pen14} for more details. We have
\[
\varphi(c) = \inf_{d\in\C}EV(c-d).
\]
Moreover, $\varphi$, $\pi^0_s$ and $\pi(\bar c,p;\cdot)$ are convex functions on $\M$ and they are nonincreasing in the directions of the {\em recession cone}
\[
\C^\infty = \{c\in\M\,|\,\bar c + \alpha c\in\C\ \forall\bar c\in\C,\ \alpha>0\}
\]
of $\C^\infty$. Clearly, $\C^\infty$ is a convex cone and $\C^\infty\subseteq\C$ with equality when $\C$ is a cone (as is the case e.g.\ under proportional transaction costs and conical constraints). We will give an explicit expression for $\C^\infty$ in terms of the market model $(S,D)$; see Corollary~\ref{cor:rec} below.

The accounting values can be bounded between the {\em super-} and {\em subhedging cost} defined by
\[
\pi^0_{\sup}(c) = \inf\{\alpha\in\reals\,|\,c-\alpha p^0\in\C\}\quad\text{and}\quad\pi_{\inf}^0(c) = \sup\{\alpha\in\reals\,|\,\alpha p^0-c\in\C\},
\]
respectively. Indeed, if $\pi^0_s(0)=0$, then, by \cite[Theorem~3.2]{pen14},
\[
\pi^0_{\inf}(c)\le\pi^0_l(c)\le\pi^0_s(c)\le\pi^0_{\sup}(c)
\]
with equalities throughout when $c-\bar\alpha p^0\in\C\cap(-\C)$ for some $\bar\alpha\in\reals$ and in this case, $\pi^0_s(c)=\bar\alpha$. Similarly, the indifference swap rates can be bounded by the {\em super-} and {\em subhedging swap rates} defined by
\[
\pi_{\sup}(p;c) = \inf\{\alpha\in\reals\,|\,c-\alpha p\in\C^\infty\}\quad\text{and}\quad\pi_{\inf}(p;c) = \sup\{\alpha\in\reals\,|\,\alpha p-c\in\C^\infty\},
\]
respectively. Indeed, if $\pi_s(\bar c,p;0)=0$, then, by \cite[Theorem~4.1]{pen14},
\[
\pi_{\inf}(p;c)\le\pi_l(\bar c,p;c)\le\pi_s(\bar c,p;c)\le\pi_{\sup}(p;c)
\]
with equalities throughout if $c-\bar\alpha p\in\C^\infty\cap(-\C^\infty)$ for some $\bar\alpha\in\reals$ and in this case, $\pi_s(\bar c,p;c)=\bar\alpha$. This last condition extends the classical condition of {\em replicability} (attainability) to nonlinear market models. In \cite[Section~2]{mad15}, elements of $\C^\infty\cap(-\C^\infty)$ are interpreted as ``liquid zero cost attainable claims''. Recall that in conical market models, $\C^\infty=\C$.

\section{Duality in optimal investment}\label{sec:ALM}

Our overall goal is to derive dual expressions for the optimum value function $\varphi$, the accounting value $\pi^0_s$ and the indifference swap rate $\pi_s(\bar c,p;\cdot)$, to analyze their properties and to relate them to more specific  instances of market models and duality theory that have appeared in the literature. We will use the functional analytic techniques of convex analysis where duality is derived from the notion of a {\em conjugate} of a convex function; see \cite{roc74}. More precisely, we apply the results of \cite{pen11c,pp12,bpp,per16}, where the general conjugate duality framework of Rockafellar was specialized to general convex stochastic optimization problems.

In order to apply conjugate duality in the setting of Section~\ref{sec:oicv}, we will assume from now on that $\M\subset L^0(\Omega,\F,P;\reals^{T+1})$ is in {\em separating duality} with another space $\Q\subset L^0(\Omega,\F,P;\reals^{T+1})$ of random sequences under the bilinear form
\[
\langle c,q\rangle := E(c\cdot q),
\]
where ``$\cdot$'' denotes the usual Euclidean inner product. We will also assume that both $\M$ and $Q$ are {\em decomposable} in the sense of \cite{roc68} and that they are closed under {\em adapted projections}. Decomposability means that $c\one_A+c'\one_{\Omega\setminus A}\in \M$ for every $c\in\M$, $c'\in L^\infty$ and $A\in\F$. The adapted projection of a $c=(c_t)_{t=0}^T\in\M$ is the $\FF$-adapted stochastic process $\ap c$ given by $\ap c_t:=E_tc_t$. Here and in what follows, $E_t$ denotes the conditional expectation with respect to $\F_t$. Recall that (see e.g.~\cite[Lemma~1]{bis73}) decomposability $\M$ and $\Q$ implies that $L^\infty\subseteq\M\subseteq L^1$ and that the pointwise maximum of two elements of $\M$ belongs to $\M$.

We will denote the adapted elements of $\M$ and $\Q$ by $\M^a$ and $\Q^a$, respectively. Since $\M$ and $\Q$ are assumed to be closed under adapted projections, $\M^a$ and $\Q^a$ are in separating duality under $(c,q)\mapsto\langle c,q\rangle$. The elements of $\Q^a$ can be interpreted as ``stochastic discount factors'' (or ``state price densities'') that extend the notion of a martingale density from classical market models. In markets without perfectly liquid cash (or other perfectly liquid numeraires), more general dual objects from $\Q^a$ are needed in order to bring out the time value of money. In a deterministic setting, the dual variables $q\in\Q^a$ represent the term structure of interest rates (prices of zero coupon bonds). In \cite[Theorem~3.4]{afp12} and \cite[Section~3]{mad15}, dual variables $q\in\Q^a$ are expressed as products of martingale densities and predictable discount processes; see also \cite[Section~3.2]{teh??}.

The convex conjugate of $\varphi:\M^a\to\ereals$ with respect to the pairing of $\M^a$ with $\Q^a$ is defined by
\[
\varphi^*(c) = \sup_{c\in\M^a}\{\langle c,q\rangle - \varphi(c)\}.
\]
The conjugate of a function on $\Q^a$ is defined analogously. If $\varphi$ is closed and proper, then $\varphi^{**}=\varphi$; see e.g.\ \cite[Theorem~5]{roc74}. In other words, we then have the dual representation
\[
\varphi(c) = \sup_{q\in\Q^a}\{\langle c,q\rangle - \varphi^*(q)\}.
\]
The maximization problem on the right is known as the {\em dual problem}. If $\varphi$ is not closed, then for some $c\in\M^a$ the dual optimum value is strictly less than $\varphi(c)$ (there is a ``duality gap''). The main result of this section gives an expression for the conjugate $\varphi^*$ in terms of the loss function $V$ and the market model $(S,D)$. Section~\ref{sec:cl} below gives sufficient conditions for the closedness of $\varphi$. 

Note that if $V$ is $P$-almost surely the indicator function of the nonpositive cone $\reals^{T+1}_-$, then the optimum value function $\varphi$ coincides with the indicator function $\delta_\C$ (of the set $\C$ of claims that can be superhedged without a cost) on $\M^a$ and its conjugate of $\varphi^*$ becomes the {\em support function}\footnote{Here one can think of $\C$ as the set of adapted claims that can be superhedged without a cost. Indeed, it is not difficult to show that the supremum is not affected if one restricts $c$ to be adapted.}
\[
\sigma_{\C}(q):=\sup_{c\in\C}\langle c,q\rangle.
\]
When $S$ is sublinear and $D$ is conical, the set $\C$ is a cone as well and $\sigma_\C$ becomes the indicator function of the {\em polar cone} 
\[
\C^* := \{q\in\Q^a\,|\,\langle c,q\rangle\le 0\ \forall c\in\C\}.
\]
In the classical perfectly liquid market model, the elements of $\C^*$ come out as the positive multiples of martingale density processes; see the remarks after Corollary~\ref{cor:support2}. Combined with the Kreps--Yan theorem, this gives a quick proof of the fundamental theorem of asset pricing.


The expression for $\varphi^*$ in Theorem~\ref{thm:varphi*} below applies to loss functions $V$ that exhibit risk aversion with respect to sequences of random payments.

\begin{assumption}\label{ass2}
$EV(\ap{c})\le EV(c)$ for all $c\in\M$.
\end{assumption}

Clearly, Assumption~\ref{ass2} holds also if $\F_t=\F$ for all $t$. It holds also if $V$ is of the form \eqref{separable} where each $V_t$ is a convex normal $\F_t$-integrand such that there exists $q\in\Q^a$ with $EV_t^*(q_t)<\infty$ for all $t$. Indeed, in that case, the inequality in Assumption~\ref{ass2} is simply Jensen's inequality for convex normal integrands; see e.g.\ \cite[Corollary~2.1]{pen11c}. 


We use the notation
\[
(\alpha S_t)(x,\omega):= 
\begin{cases}
  \alpha S_t(x,\omega)\quad&\text{if } \alpha >0\\
  \delta_{\cl\dom S_t(\cdot,\omega)}(x)\quad&\text{if }\alpha=0.
\end{cases}
\]
The definition in the case $\alpha=0$ is motivated by the fact that the function $\alpha S_t(\cdot,\omega)$ is the {\em epi-graphical limit} of $\alpha^\nu S_t(\cdot,\omega)$ as $\alpha^\nu\downto 0$; this follows from \cite[Proposition 7.29]{rw98} and the property that proper closed convex functions are uniformly bounded from below on bounded sets. By \cite[Proposition 14.44 and Proposition 14.46]{rw98}, $\alpha S$ is a convex normal integrand for every measurable $\alpha\ge 0$.

We will denote the space of adapted integrable $\reals^J$-valued processes by $\N^1$. In order to simplify the notation in the statements below, we fix a dummy variable $w_{T+1}\in L^1(\Omega,\F,P;\reals^J)$. Since $D_T\equiv\{0\}$, we have $\sigma_{D_T}\equiv 0$, so the value of $w_{T+1}$ does not affect any of the expressions. The proof of the following will be given in Section~\ref{sec:adapted}.

\begin{theorem}\label{thm:varphi*}
If $V$ satisfies Assumptions~\ref{ass1} and \ref{ass2} and if $S(x,\cdot)\in\M^a$ for all $x\in\reals^J$, then the conjugate of the optimum value function $\varphi$ of \eqref{alm} can be expressed as
\[
\varphi^*(q) = EV^*(q) + \inf_{w\in\N^1}E\left\{\sum_{t=0}^T\left[(q_tS_t)^*(w_t) + \sigma_{D_t}(E_t\Delta w_{t+1})\right]\right\},
\]
where the infimum is attained for every $q\in\Q^a$.
\end{theorem}

The assumption that $S(x,\cdot)\in\M^a$ for $x\in\reals^J$ is a generalization of the assumption $S(x,\cdot)\in L^1$ made in \cite{pen11}. Recall that $\M^a\subseteq L^1$. We will see in Section~\ref{sec:oc} that a stronger assumption on $S$ may allow one to establish the existence of dual solutions and the necessity of optimality conditions for \eqref{alm}. The assumption that $S(x,\cdot)\in\M^a$ is close also to \cite[Assumption~3.1]{bc11} which, in the perfectly liquid case, asks the price process to belong ``locally'' to an Orlicz space associated with the utility function being optimized. 

When $V(\cdot,\omega)=\delta_{\reals^{T+1}_-}$ for $P$-almost all $\omega$, Theorem~\ref{thm:varphi*} gives the following is a variant of \cite[Lemma~A.1]{pen11}.

\begin{corollary}\label{cor:support}
If $S(x,\cdot)\in\M^a$ for all $x\in\reals^J$, then
\[
\sigma_{\C}(q) = \inf_{w\in\N^1}E\left\{\sum_{t=0}^T[(q_tS_t)^*(w_t) + \sigma_{D_t}(E_t\Delta w_{t+1})]\right\}
\]
for $q\in\M^a_+$ while $\sigma_{\C}(q)=+\infty$ otherwise. Moreover, the infimum is attained for every $q\in\M^a$. 
\end{corollary}

\begin{proof}
When $V(\cdot,\omega)=\delta_{\reals^{T+1}_-}$, Assumptions~\ref{ass1} and \ref{ass2} clearly hold, and we have $\varphi=\delta_{\C}$ so that $\varphi^*=\sigma_{\C}$. Since $V^*(\cdot,\omega)=\delta_{\reals^{T+1}_+}$, Theorem~\ref{thm:varphi*} gives
\[
\varphi^*(q) = \inf_{w\in\N^1}E\left\{\sum_{t=0}^T[(q_tS_t)^*(w_t) + \sigma_{D_t}(E_t\Delta w_{t+1})]\right\}
\]
for $q\in\M^a_+$ and $\varphi^*(q)=+\infty$ for $q\notin\M^a_+$.
\end{proof}

When $S$ is sublinear and $D$ is conical, we have $S^*_t=\delta_{\dom S^*_t}$ and $\sigma_{D_t}=\delta_{D^*_t}$. If in addition, $S(x,\cdot)\in\M^a$ for all $x\in\reals^J$, so that $\dom S_t=\reals^J$ and $(q_t S_t)^*(w)=\delta(w\,|\,q_t\dom S^*_t)$, Corollary~\ref{cor:support} gives the following expression for the polar cone
\[
\C^* =\{q\in\Q^a_+\,|\, \exists w\in\N^1:\ w_t\in q_t\dom S^*_t,\ E_t[\Delta w_{t+1}]\in D_t^*\ P\text{-a.s.}\}.
\]
The polar cone of $\C$ can also be written as\footnote{Here we use the fact that every selector $w$ of $\dom S^*_t$ is integrable. Indeed, by Fenchel inequality, $|w|\le \sup_{|x|\le 1}S_t(x)+S^*_t(w)$, where the supremum is integrable since $S(x,\cdot)\in\M^a$.}
\[
\C^* =\{q\in\Q^a_+\,|\, \exists s\in\N^1:\ s_t\in\dom S^*_t,\ E_t[\Delta(q_{t+1}s_{t+1})]\in D_t^*\ P\text{-a.s.}\}.
\]
In unconstrained linear models with $S_t(x,\omega)=s_t(\omega)\cdot x$ and $D_t(\omega)=\reals^J$, we have $\dom S^*_t=\{s_t\}$ and $D_t^*(\omega)=\{0\}$ so that
\[
\C^* = \{q\in\Q^a_+\,|\, \text{$qs$ is a martingale}\}
\]
in consistency with \cite[Definition~2.2]{teh??} in the context of numeraire-free linear models. In unconstrained models with bid-ask spreads, $\dom S^*_t$ is the cube $[s^-_t,s^+_t]$ and
\[
\C^* = \{q\in\Q^a_+\,|\, \exists s\in\N:\ s^-_t\le s_t\le s^+_t,\ qs\ \text{is a martingale}\}.
\]
The processes $s$ above correspond to ``shadow prices'' in the sense of \cite{cms14}; see Section~\ref{sec:oc} for more details on this. 
When short selling is prohibited, i.e.\ when $D_t=\reals^J_+$, the condition $E_t[\Delta(q_{t+1}s_{t+1})]\in D_t^*$ means that $qs$ is a supermartingale. In models with perfectly liquid cash (see Example~\ref{ex:liquid}), the elements of $\C^*$ can be expressed as
\[
\C^* = \{q\in\Q^a_+\,|\, \exists Q\in\P,\ \alpha\ge 0: q_t=\alpha E_t\frac{dQ}{dP}\}
\]
where $\P = \{Q\ll P\,|\,\exists\tilde s\in\tilde\N:\ \tilde s_t\in\dom\tilde S^*_t,\ E^Q_t\Delta\tilde s_{t+1}\in\tilde D^*_t\ Q\text{-a.s.}\}$ and $\tilde\N$ denotes the set of adapted $\reals^{J\setminus\{0\}}$-valued processes. In unconstrained models, $\P$ is the set of absolutely continuous probability measures $Q$ that admit martingale selectors of $\dom S^*_t$. In the classical linear model with $S_t(x,\omega)=s_t(\omega)\cdot x$ we simply have $\dom S_t^*=\{s_t\}$ so the set $\P$ becomes the set of absolutely continuous martingale measures for $s$.

Using Corollary~\ref{cor:support}, we can write Theorem~\ref{thm:varphi*} more compactly as follows.

\begin{corollary}\label{cor:support2}
Under the assumptions of Theorem~\ref{thm:varphi*}, $\varphi^* = EV^* + \sigma_{\C}$.
\end{corollary}

Corollary~\ref{cor:support2} was given in \cite[Lemma~3]{pen14b} under more restrictive assumptions on the loss function. Moreover, the proof of \cite[Lemma~3]{pen14b} was based on \cite[Theorem~2.2]{pen11c}, which suffers from errors that were corrected in \cite{bpp}; see \cite{pen16}.

\section{Closedness of the value function}\label{sec:cl}

This section gives conditions for the closedness of $\varphi$ and thus, the validity of the dual representation $\varphi=\varphi^{**}$. The closedness conditions will be given in terms of asymptotic behavior of the market model and the loss function. The asymptotic behavior of a convex function $h$ is described by its {\em recession function} $h^\infty$ which can be expressed as
\[
h^\infty(x)=\sup_{\alpha>0}\frac{h(\alpha x+\bar x)-h(\bar x)}{\alpha},
\]
where the supremum is independent of the choice of $\bar x\in\dom h$; see \cite[theorem~8.5]{roc70a}. The recession function of a convex function is convex and positively homogeneous. 


Given a market model $(S,D)$, the function $S_t^\infty(\cdot,\omega):=S_t(\cdot,\omega)^\infty$ is a convex $\F_t$-integrand and $D_t^\infty(\omega):=D_t(\omega)^\infty$ is an $\F_t$-measurable set; see \cite[Exercises~14.54 and 14.21]{rw98}. This way, every $(S,D)$ gives rise to a conical market model $(S^\infty,D^\infty)$. Similarly, a convex loss function $V$ gives rise to a sublinear loss function $V^\infty(\cdot,\omega):=V(\cdot,\omega)^\infty$. Note that the growth properties in Assumption~\ref{ass1} mean that $V^\infty(c,\omega)\le 0$ for all $c\in\reals^{T+1}_-$ and $\{c\in\reals^{T+1}_+\,|\, V^\infty(c,\omega)\le 0\}=\{0\}$. Our closedness result requires the following.

\begin{assumption}\label{ass3}
$\{x\in\N_{D^\infty}\,|\, V^\infty(S^\infty(\Delta x))\le 0\}$ is a linear space.
\end{assumption}

If the loss function $V$ satisfies
\begin{equation}\label{assV}
\{c\in\reals^{T+1}\,|\,V^\infty(c,\omega)\le 0\}=\reals^{T+1}_-,
\end{equation}
then Assumption~\ref{ass3} means that $\{x\in\N_{D^\infty}\,|\, S^\infty(\Delta x)\le 0\}$ is a linear space. This generalizes the classical no-arbitrage condition from perfectly liquid market models to illiquid ones. In particular, when $D\equiv\reals^J$ (no portfolio constraints), the linearity condition becomes a version of the ``robust no-arbitrage condition'' introduced in \cite{sch4} for the currency market model of \cite{kab99}; see \cite{pen12} for the details and further examples that go beyond no-arbitrage conditions. The linearity of $\{x\in\N_{D^\infty}\,|\, S^\infty(\Delta x)\le 0\}$ can be seen as a relaxation of the condition of ``efficient friction'', which, in the present setting, would require this set to reduce to the origin; see \cite{bcr15} for a review of the efficient friction condition and an extension to a models with dividend payments.

Condition \eqref{assV} clearly implies Assumption~\ref{ass1}. It holds in particular if $V$ is of the form \eqref{separable}, where $V_t(\cdot,\omega)$ are univariate nonconstant loss functions with $V_t^\infty\ge 0$. Indeed, since $V_t$ are nondecreasing, we have $V_t^\infty\le 0$ on $\reals_-$ while the nonconstancy implies $V_t^\infty(c_t,\omega)>0$ for $c_t>0$; see \cite[Corollary~8.6.1]{roc70a}. 
Condition~\eqref{assV} certainly holds if $V_t$ are differentiable and satisfy the Inada conditions
\[
\lim_{c\downto-\infty}V_t'(c,\omega)=0\quad\text{and}\quad\lim_{c\upto\infty}V_t'(c,\omega)=+\infty
\]
since then, $V_t^\infty\equiv\delta_{\reals_-}$.

We will need one more assumption on the loss functions.

\begin{assumption}\label{scale}
There exists $\lambda\ne 1$ such that
\[
\lambda\dom EV^*\subseteq \dom EV^*.
\]
\end{assumption}

Assumption~\ref{scale} holds, in particular, if $V$ is bounded from below by an integrable function since then $0\in\dom EV^*$. Lemma~\ref{lem:rrae} below gives more general conditions. The conditions extend the {\em asymptotic elasticity} conditions from \cite{ks99} and \cite{sch1} to multivariate, random and possibly nonsmooth loss functions. For a univariate loss function $v$, the conditions can be stated as
\begin{align}
\exists \lambda>1,\, \bar y\in\dom Ev^*,\, C\ge 0:\ v^*(\lambda y)&\le Cv^*(y)\quad\forall\, y\ge \bar y,\tag{RAE$_+$}\label{rae+}\\
\exists  \lambda\in(0,1),\, \bar y\in\dom Ev^*,\, C>0:\ v^*(\lambda y)&\le Cv^*(y)\quad\forall\, y\le\bar y \tag{RAE$_-$}\label{rae-}.
\end{align}
Various reformulations of the asymptotic elasticity conditions given for deterministic functions in \cite{ks99,sch1} are extended to random and nonsmooth loss functions in the appendix. The following lemma gives multivariate extensions.

\begin{lemma}\label{lem:rrae}
If $v^*(\eta,\omega):=V^*(\eta y(\omega),\omega)$ satisfies \eqref{rae+} for every $y\in\dom EV^*$, then
\[
\lambda\dom EV^*\subseteq\dom EV^*\quad\forall\lambda\ge 1.
\]
If $v^*(\eta,\omega)$ satisfies \eqref{rae-} for every $y\in\dom EV^*$, then
\[
\lambda\dom EV^*\subseteq\dom EV^*\quad\forall\lambda\in(0,1].
\]    
\end{lemma}

\begin{proof}
Given $y\in\dom EV^*$, the normal integrand $h=v^*$ satisfies in both cases the assumptions of Lemma~\ref{lem:scale} in the appendix with $\one\in\dom Eh$.
\end{proof}

\begin{remark}
In the univariate case, the radial conditions in Lemma~\ref{lem:rrae} are satisfied by loss functions that satisfy the corresponding asymptotic elasticity condition. Indeed, if $V^*$ satisfies \eqref{rae+} and $y\in\dom EV^*$, then $v^*(\eta,\omega):=V^*(\eta y(\omega),\omega)$ satisfies \eqref{rae+} with the same $\lambda$ and $C$ but with $\bar y$ replaced by $\one_{\{y>0\}}\bar y/y$. Similarly for \eqref{rae-}.
\end{remark}

We are now ready to state the main result of this section. Besides closedness of the optimum value function $\varphi$ and the existence of optimal trading strategies, it gives an explicit expression for the recession function of $\varphi$. These results turn out to be useful in the analysis of valuations of contingent claims in Section~\ref{sec:val}. The proof of the following will be given in Section~\ref{sec:adapted}.

\begin{theorem}\label{thm:cl}
Let Assumptions~\ref{ass1}, \ref{ass3} and \ref{scale} hold and assume that
\[
EV^*(q) + \inf_{w\in\N^1}E\left\{\sum_{t=0}^T[(q_tS_t)^*(w_t) + \sigma_{D_t}(E_t\Delta w_{t+1})]\right\}<\infty
\]
for some $q\in\Q$. Then $\varphi$ is $\sigma(\M^a,\Q^a)$-closed in $\M^a$, the infimum in \eqref{alm} is attained for all $c\in\M^a$ and 
\begin{align*}
\varphi^\infty(c) &= \inf_{x\in\N_D}EV^\infty(S^\infty(\Delta x)+c).
\end{align*}
\end{theorem} 


Under the assumptions of Theorem~\ref{thm:varphi*}, the last assumption in Theorem~\ref{thm:cl} holds, in particular, if $\varphi^*$ is proper. The assumption that $\varphi^*$ be proper holds automatically if $V$ is bounded from below by an integrable function since then $\varphi$ is bounded as well so that $\varphi^*(0)<\infty$. If $V$ is unbounded from below, then the properness of $\varphi^*$ is a more delicate matter. In the classical linear model, this holds if there exists a martingale measure $Q\ll P$ with a density process $q\in\Q^a$ with $EV^*(q)<\infty$. This is analogous to Assumption~1 of \cite{sch3} where the market was assume perfectly liquid and the dual problem was formulated over equivalent martingale measures.
In general, properness of $\varphi^*$ means that $\varphi$ is bounded from below on a Mackey-neighborhood\footnote{The {\em Mackey-topology} on $\M^a$ is the strongest locally convex topology compatible with the pairing with $\Q^a$.} $U$ of the origin. Indeed, this means that the lower semicontinuous hull of $\varphi$ is finite at the origin and then $\varphi^*$ is proper, by \cite[Theorem~4]{roc74}. This certainly holds if $\varphi$ is finite and continuous at the origin. Sufficient conditions for that will be given in Section~\ref{sec:oc}.

The following result summarizes our findings so far. It gives a dual representation for the optimum value function of \eqref{alm} by combining Theorems~\ref{thm:varphi*}, \ref{thm:cl} and Corollary~\ref{cor:support} with the general biconjugate theorem for convex functions.

\begin{theorem}\label{thm:dr}
If Assumptions~\ref{ass1}-\ref{scale} hold, $S(x,\cdot)\in\M^a$ for all $x\in\reals^J$ and $\varphi^*$ is proper, then
\[
\varphi(c) = \sup_{q\in\Q^a}\left\{\langle c,q\rangle - EV^*(q) - \sigma_\C(q)\right\},
\]
where $\sigma_\C$ is given by Corollary~\ref{cor:support}.
\end{theorem}

The following example specializes Theorem~\ref{thm:dr} to conical market models with perfectly liquid cash.

\begin{example}
Consider Example~\ref{ex:liquid} and assume that $D$ is conical, $S$ is sublinear, and that $\{x\in\N_{D}\,|\, S(\Delta x)\le 0\}$ is a linear space (which holds, in particular, under the robust no arbitrage condition if there are no constraints; see \cite[Section~4]{pen12}). If $V_T(\cdot,\omega)$ is nonconstant, convex loss function satisfying $V^\infty\ge 0$ and either \eqref{rae+} or \eqref{rae-}, then
\begin{align*}
\varphi(c) 
&= \sup_{\lambda\ge 0}\sup_{Q\in\P}\left\{\lambda E^Q\sum_{t=0}^Tc_t-EV_T^*(\lambda E_T\frac{dQ}{dP})\right\}
\end{align*}
(see the representations of $\C^*$ at the end of Section~\ref{sec:ALM}). This is analogous to the dual problems derived for continuous time models e.g.\ in 
\cite{oz9} or \cite{bfg11}. While these references studied optimal investment with random endowment in perfectly liquid markets, the above applies to illiquid markets in discrete time. In the exponential case
\[
V_T(c) = \frac{1}{\alpha}(e^{\alpha c}-1),
\]
the supremum over $\lambda>$ is easily found analytically and one gets
\begin{align*}
\varphi(c) 
&= \frac{1}{\alpha}\exp\left[\sup_{Q\in\P}\{\alpha E^Q\sum_{t=0}^Tc_t -H(Q|P)\}\right]-\frac{1}{\alpha},
\end{align*}
where $H(Q|P)$ denotes the entropy of $Q$ relative to $P$. This is an illiquid discrete-time version of the dual representation obtained in \cite{dgrsss}; see also \cite{bec3}.
\end{example}

Much of duality theory in optimal investment has studied the optimum value as a function of the initial endowment only; see e.g.\ \cite{ks99} or \cite{kr7}. 

\begin{example}[Indirect utilities]
The function $v(c_0)=\varphi(c_0,0,\ldots,0)$ gives the optimum value of \eqref{alm} for an agent with initial capital $-c_0$ and no future liabilities/endowments. We can express $v$ as
\begin{equation}\label{V}
v(c_0) = \inf_{d\in\C(c_0)}EV(d),
\end{equation}
where $\C(c_0) = \{d\in\M^a\,|\,(c_0,0,\ldots,0)-d\in\C\}$ is the set of future endowments needed to risklessly cover an initial payment of $c_0$. If $\varphi$ is proper and lower semicontinuous (see Theorem~\ref{thm:cl}), the biconjugate theorem gives
\[
v(c_0) = \sup_{q\in\Q^a}\{c_0q_0 - \varphi^*(q)\} = \sup_{q_0\in\reals}\{c_0q_0-u(q_0)\} 
\]
where
\[
u(q_0) = \inf_{z\in\Q^a}\{\varphi^*(z)\,|\,z_0=q_0\}.
\]
If $\C$ is conical, we can use Corollary~\ref{cor:support2} to write $u$ analogously to \eqref{V} as
\[
u(q_0) = \inf_{z\in\Y(q_0)}EV^*(z),
\]
where $\Y(q_0)=\{z\in\C^*\,|\, z_0=q_0\}$. Even if $\varphi$ is closed, there is no reason to believe that $u$ would be lower semicontinuous as well nor that the infimum in its definition is attained. In some cases, however, it is possible to enlarge the set $\Y(q_0)$ so that the function $u$ becomes lower semicontinuous and the infimum is attained; see \cite{cms14} for a sublinear two asset model without constraints.
\end{example}

When $V=\delta_{\reals^{T+1}_-}$, the assumptions of Theorem~\ref{thm:cl} are automatically satisfied so we obtain the following result on the set $\C$ of claims that can be superhedged without a cost.

\begin{corollary}\label{cor:rec}
If $\{x\in\N_{D^\infty}\,|\,S^\infty(\Delta x)\le 0\}$ is a linear space, then $\C$ is closed and its recession cone can be expressed as
\[
\C^\infty=\{c\in\M^a\,|\,\exists x\in\N_{D^\infty}:\ S^\infty(\Delta x)+c\le 0\}.
\]
\end{corollary}

The first part of Corollary~\ref{cor:rec} was given in \cite[Section~4]{pen12}, where it was also shown that the linearity condition is implied by the ``robust no-arbitrage condition'', which reduces to the classical no-arbitrage condition in the classical perfectly liquid market model. Combined with the Kreps--Yan theorem, this gives a quick proof of the famous Dalang--Morton--Willinger theorem.

The recession cone $\C^\infty$ of $\C$ plays an important role in contingent claim valuation; see \cite{pen14} and Section~\ref{sec:val} below. In particular, the indifference swap rates are uniquely determined by the market model (and are thus independent of the agent's preferences $V$ and financial position $\bar c$) when $c-\alpha p\in\C^\infty\cap(-\C^\infty)$ for some $\alpha\in\reals$; see \cite[Theorem~4.1]{pen14}. This generalizes the classical notion of ``attainability'' to illiquid market models. Accordingly, in general market models, the notion of completeness (attainability of all $c\in\M^a$) extends to the property of $\C^\infty\cap(-\C^\infty)$ being a maximal linear subspace of $\M^a$. The maximality means that if $p\notin\C^\infty\cap(-\C^\infty)$, then the linear span of $p\cup[\C^\infty\cap(-\C^\infty)]$ is all of $\M^a$. Under the conditions of Theorem~\ref{thm:cl} and the mild additional condition that $V^\infty\ge 0$, the condition $p\notin\C^\infty$, turns out to be necessary and sufficient for $\pi_s(\bar c,p;\cdot)$ to be a proper lsc function on $\M^a$; see Theorem~\ref{thm:idsr} below.

\section{Optimality conditions and shadow prices}\label{sec:oc}

Under the assumptions of Theorems~\ref{thm:varphi*} and \ref{thm:cl}, the optimum value of \eqref{alm} equals that of the {\em dual} problem
\[
\minimize_{q\in\Q^a, w\in\N^1}\quad E\left\{V^*(q) + \sum_{t=0}^T[(q_tS_t)^*(w_t) + \sigma_{D_t}(E_t\Delta w_{t+1}) - c_tq_t]\right\}.
\]
The optimal $q\in\Q^a$, if any exist, are characterized by the equality $\varphi(c)+\varphi^*(q)=\langle c,q\rangle$ which means that $q$ is a {\em subgradient} of $\varphi$ at $c$, i.e.
\[
\varphi(c')\ge\varphi(c)+\langle c'-c,q\rangle\quad\forall c'\in\M^a.
\]
Optimal dual solutions $q\in\Q^a$ can thus be interpreted as {\em marginal prices} for the claims $c\in\M^a$. In particular, if $\varphi$ happens to be Gateaux differentiable at $c$, then $q$ is the derivative of $\varphi$ at $c$. We will denote the {\em subdifferential}, i.e.\ the set of subgradients, of $\varphi$ at $c\in\M^a$ by $\partial\varphi(c)$. 

Much like in the conjugate duality framework of \cite{roc74}, the dual solutions allow us to write down optimality conditions for the solutions of the primal problem \eqref{alm}. The proof of the following will be given in Section~\ref{sec:adapted}.

\begin{theorem}\label{thm:oc}
Let $c\in\M^a$. If Assumptions~\ref{ass1}--\ref{scale} hold and $(w,q)\in\Q^a\times\N^1$ solves the dual, then an $x\in\N$ solves \eqref{alm} if and only if it is feasible and
\begin{align*}
E_t\Delta w_{t+1} &\in N_{D_t}(x_t),\\
q &\in\partial V(S(\Delta x)+c),\\
w_t&\in\partial(q_tS_t)(\Delta x_t).
\end{align*}
Conversely, if $x$ and $(w,q)$ are feasible primal-dual pair satisfying the above system, then $x$ solves the primal, $(w,q)$ solves the dual, and $\varphi$ is closed at $c$.
\end{theorem}

In the classical linear model where $S_t(x,\omega)=s_t(\omega)\cdot x$, the optimality conditions in Theorem~\ref{thm:oc} simplify to
\begin{align*}
E_t\Delta (q_{t+1}s_{t+1}) &\in N_{D_t}(x_t),\\
q &\in\partial V(s\cdot\Delta x+c).
\end{align*}
In the unconstrained case where $D\equiv\reals^J$, the first condition means that $qs$ is a martingale. 
On the other hand, we can write the last condition as $w_t=q_t\bar s_t$, where $\bar s_t\in L^0(\Omega,\F_t,P;\reals^J)$ is such that
\[
\bar s_t\in\partial S_t(\Delta x_t).
\]
Following \cite{cms14}, we call such a process $\bar s$ a {\em shadow price}. By \cite[Theorem~23.5]{roc70a}, we can write this as $\Delta x_t\in\partial S^*_t(\bar s_t)$. If $S$ is positively homogeneous like in \cite{cms14}, this becomes
\[
\Delta x_t\in N_{\dom S^*_t}(\bar s_t).
\]
In particular then, the optimal strategy trades only when $\bar s_t\notin\inte\dom S^*_t$ and in that case, the increment $\Delta x_t$ belongs to the (outward) normal to $\dom S^*_t$ at $\bar s_t$. This extends the complementarity condition from \cite[Definition~2.2]{cms14} to multiple assets and general sublinear trading costs. If, moreover, $S^*(\bar s)\in\M^a$, and $\bar x$ is optimal in \eqref{alm}, then it is optimal also in the linearized problem
\[
\minimize\quad E V(\bar s\cdot\Delta x+\bar c)\quad\ovr\quad x\in\N_D,
\]
where $\bar c:=c-S^*(\bar s)$. Indeed, if $\bar x$, $(\bar w,\bar q)$ are optimal in \eqref{alm} and the dual, we have $S(\Delta\bar x)+S^*(\bar s)=\bar s\cdot\Delta\bar x$, by the optimality conditions in Theorem~\ref{thm:oc}, so they satisfy the optimality conditions for the linearized model for which they are feasible as well, so the claim follows from the last part of Theorem~\ref{thm:oc}. 

The condition $\Delta x\in N_{\dom S^*_t}(\bar s_t)$ is closely related to the result of \cite{dn90}, who found that under transaction costs, there is a ``no-transaction region'' where optimal trading strategies stay constant. In the sublinear case, this region is the interior of $\dom S^*_t$. Note also that if $S$ happens to be strictly convex, $\partial S^*$ is (at most) single-valued (see \cite[Theorems~26.1 and~26.1]{roc70a}), so the shadow price characterizes the optimal trading strategy uniquely.




Section~3 of \cite{cms14} gives an example where shadow prices do not exist and thus, the supremum in the dual representation of $\varphi$ is not attained. The following gives sufficient conditions for the attainment. Recall that the {\em Mackey-topology} on $\M^a$ is the strongest locally convex topology compatible with the pairing with $\Q^a$.

\begin{theorem}\label{thm:bounded}
Under the assumptions of Theorem~\ref{thm:varphi*}, the dual optimum is attained if $\varphi$ is bounded from above on a Mackey-neighborhood of $c$.
\end{theorem}

\begin{proof}
By Theorem~\ref{thm:varphi*}, dual solutions coincide with subgradients of $\varphi$ at the origin. The claim thus follows from \cite[Theorem~11]{roc74}.
\end{proof}


Much like in \cite[Corollary~5.2]{bc11}, the dual attainment in $\Q^a$ can be guaranteed under appropriate continuity assumptions on the expected loss function $EV$. Recall that a locally convex topological vector space is {\em barreled} if every closed convex absorbing set is a neighborhood of the origin. By \cite[Corollary 8B]{roc74}, a lower semicontinuous convex function on a barreled space is continuous throughout the algebraic interior (core) of its domain. On the other hand, by \cite[Theorem 11]{roc74}, continuity implies subdifferentiability. Fr\'echet spaces and, in particular, Banach spaces are barreled. We will say that $\M^a$ is barreled if it is barreled with respect to a topology compatible with the pairing with $\Q^a$.

\begin{example}
Assume that $\M^a$ is barreled. Then the boundedness condition in Theorem~\ref{thm:bounded} is satisfied, if $EV$ is finite throughout $\M^a$ and $EV^*$ is proper on $\Q^a$. The finiteness of $EV$ fails in Example~\ref{ex:liquid}, but the boundedness condition holds if $EV_T(c_0+\cdots+c_T)<\infty$ for all $c\in\M^a$ and either $\varphi$ or $c\mapsto EV_T(c_0+\cdots+c_T)$ is lower semicontinuous on $\M^a$ (see Theorem~\ref{thm:cl}). 
\end{example}

\begin{proof}
By \cite[Theorem~21]{roc74}, $EV$ is lower semicontinuous so, by \cite[Corollary~8B]{roc74}, it is continuous on $\M^a$. Choosing $x=0$, gives
\[
\varphi(c) \le EV(c),
\]
so \cite[Theorem~8]{roc74} implies that $\varphi$ is continuous and thus subdifferentiable throughout $\M^a$, by \cite[Theorem~11]{roc74}. Similarly, in Example~\ref{ex:liquid},
\[
\varphi(c) \le EV_T(c_0+\cdots+c_T),
\]
so, by \cite[Corollary~8B]{roc74}, both of the conditions imply the continuity of $\varphi$.
\end{proof}

\section{Duality in contingent claim valuation}\label{sec:val}

The main results of \cite{pen14} relate the accounting values and indifference swap rates to arbitrage bounds and the classical replication based values. Section~6 of \cite{pen14} gives conditions on lower semicontinuity and properness of $\pi^0_s$ and $\pi_s(\bar c,p;\cdot)$. This section refines those conditions and gives dual expressions for $\pi^0_s$ and $\pi_s(\bar c,p;\cdot)$. We start with accounting values.

\subsection{Accounting values}

As noted in the introduction, the accounting value $\pi^0_s$ extends the notion of a convex risk measure to sequences of payments and markets without a perfectly liquid cash-account. This section extends the analogy by giving dual representations for $\pi^0_s$ which reduce to the well-known dual representations of risk measures when applied in the single period setting with perfectly liquid cash. We also give general conditions under which the conjugate (``penalty function'') in the dual representation separates into two components, one corresponding to the market model and the other one to the agent's risk preferences. This can be seen as an extension of a corresponding separation in models with perfectly liquid cash; see e.g.~\cite[Proposition~4.104]{fs11}.

Note that the accounting value can be expressed as
\begin{align*}
\pi^0_s(c) &= \inf\{\alpha\,|\, c-\alpha p_0\in\A\},
\end{align*}
where $\A=\{c\in\M^a\,|\,\varphi(c)\le 0\}$ consists of financial positions which the agent can cover with acceptable level of risk (as measured by $EV$) given the possibility to trade in the illiquid markets described by $S$ and $D$. This is analogous to the correspondence between convex risk measures and their acceptance sets in \cite{adeh99} where the financial market is described by a single perfectly liquid asset. Besides the market model, another notable extension here is that acceptance sets consist of sequences of payments instead of payments at a single date. 

The dual representation for $\pi^0_s$ below involves the support function of $\A$ which corresponds to the ``penalty function'' in the dual representation of a convex risk measure; see \cite[Chapter~4]{fs11} and the references therein. Under mild conditions the support function separates into two terms: the first term is the support function of the set
\[
\B = \{c\in\M^a\,|\,EV(c)\le 0\}
\]
while the second one is the support function of the set $\C$ of claims that can be superhedged without a cost. While $\B$ depends only on the agent's risk preferences, $\C$ depends only on the market model. 

\begin{theorem}\label{thm:acc}
If $\F_0$ is the trivial sigma-field, $V^\infty\ge 0$ and the assumptions of Theorem~\ref{thm:cl} hold, then the conditions
\begin{enumerate}
\item\label{plsc}
$\pi^0_s$ is proper and lower semicontinuous,
\item\label{p0}
$\pi^0_s(0)>-\infty$,
\item\label{rec}
$p^0\notin\C^\infty$,
\item\label{dual}
$q_0=1$ for some $q\in\dom\sigma_\C$
\end{enumerate}
are equivalent and imply the validity of the dual representation
\[
\pi^0_s(c) = \sup_{q\in\Q^a}\left\{\left.\langle c,q\rangle - \sigma_\A(q)\,\right|\,q_0=1\right\}.
\]
If $\inf\varphi<0$, then under the assumptions of Theorem~\ref{thm:varphi*}, $\sigma_\A = \sigma_\B + \sigma_\C$ and
\[
\sigma_\B(q) = \inf_{\alpha>0}\alpha EV^*(q/\alpha).
\]
\end{theorem}

\begin{proof}
Closedness of $\varphi$ in Theorem~\ref{thm:cl} implies the closedness of $\A$. By Lemma~\ref{lem:piD} in the appendix, the dual representation is then valid under the first two conditions which are both equivalent to $p^0\notin\A^\infty$ which in turn is equivalent to the existence of a $q\in\dom\sigma_\A$ with $\langle p^0,q\rangle=1$. Here, $\langle p^0,q\rangle=q_0$ since $p^0=(1,0,\ldots,0)$ and $\F_0$ is the trivial sigma field by assumption. By \cite[Corollary~6B]{roc66}, $\A^\infty=\{c\in\M^a\,|\,\varphi^\infty(c)\le 0\}$. When $V^\infty\ge 0$, the expression for $\varphi^\infty$ in Theorem~\ref{thm:cl} yields
\[
\A^\infty=\{c\in\M^a\,|\,\exists x\in\N_{D^\infty}:\ S^\infty(\Delta x)+c\le 0\},
\]
which, by Corollary~\ref{cor:rec}, equals $\C^\infty$. Thus, the first two conditions are both equivalent to 3. Another application of Lemma~\ref{lem:piD} now implies that 3 is equivalent to 4. This finishes the proof of the first claim.

By Lagrangian duality (see e.g.\ Example~1'' on page 45 of \cite{roc74}), the condition $\inf\varphi<0$ implies
\begin{align*}
\sigma_\A(q) &= \sup_{c\in \M^a}\{\langle c,q\rangle\,|\,\varphi(c)\le 0\}\\
&= \inf_{\alpha>0}\sup_{c\in \M^a}\{\langle c,q\rangle-\alpha\varphi(c)\}\\
&= \inf_{\alpha>0}\alpha\varphi^*(q/\alpha)\\
&= \inf_{\alpha>0}\alpha EV^*(q/\alpha) + \sigma_\C(q),
\end{align*}
where the last equality comes from Corollary~\ref{cor:support2} and the fact that $\sigma_\C$ is positively homogeneous. By Lagrangian duality again,
\begin{align*}
\sigma_\B(q) &= \sup_{c\in\M^a}\{\langle c,q\rangle\,|\,EV(c)\le 0\}.
\end{align*}
Under Assumption~\ref{ass2}, we have for all $q\in\Q^a$
\begin{align*}
\sup_{c\in\M^a}\{\langle c,q/\alpha\rangle- EV(c)\} &= \sup_{c\in L^1}\{\langle c,q/\alpha\rangle- EV(c)\}\\ 
&=EV^*(q/\alpha),
\end{align*}
where the last equality comes from the interchange rule \cite[Theorem~14.60]{rw98}.
\end{proof}

The first part of Theorem~\ref{thm:acc} is analogous to Corollary~1 in \cite[Section~4.3]{fkm15} which was concerned with risk measures in a general single period setting. The dual representation in the second part is analogous to the one in \cite[Proposition~4.104]{fs11} in the single period setting with portfolio constraints and linear trading costs. See also \cite[Theorem~3.6]{be5}, which gives a dual representation for the infimal convolution of two convex risk measures.

When $\C$ is a cone, the dual representation in Theorem~\ref{thm:acc} can be written as
\[
\pi^0_s(c) = \sup\left\{\left.\langle c,q\rangle - \sigma_\B(q)\,\right|\,q\in\C^*,\ q_0=1\right\}.
\]
If, in addition, $EV$ is sublinear, then $\B=\{c\in\M^a\,|\,EV(c)\le 0\}$ is a cone and $\sigma_\B=\delta_{\B^*}$ (the indicator function of the polar cone), so we get the more familiar expression
\begin{equation}\label{gdb}
\pi^0_s(c) = \sup\left\{\left.\langle c,q\rangle \,\right|\,q\in\C^*\cap\B^*,\ q_0=1\right\}.
\end{equation}
In the completely risk averse case where $V=\delta_{\reals^{T+1}_-}$, we have $\B=\M^a_-$ and $\B^*=\Q^a_+$. Since $\C^*\subset\Q^a_+$, we thus recover the dual representation of the superhedging cost; see e.g.\ \cite{pen12}. In general, $\B$ may be interpreted much like the set of ``desirable claims'' in the theory of good-deal bounds; see \cite{ch2} and the references therein. In the dual representation above, the polar cone $\B^*$ restricts the set of stochastic discount factors used in the valuation of the claim $c$, thus making the value lower than the superhedging cost. This is simply a dual formulation of the no-arbitrage bound obtained by purely algebraic arguments in \cite{pen14}; see Section~\ref{sec:oicv}. The general structure is similar also to the models of ``two-price economies'' in \cite{ebe15}, \cite{mad15} and their references. In particular, expression \eqref{gdb} has the same form as the dual representation of the ``ask price'' in \cite[Section~3]{mad15} which addresses conical markets without perfectly liquid numeraire in the context of finite probability spaces.

While Theorem~\ref{thm:acc} is formulated for short positions, it can be immediately translated to accounting values for long positions through the identity $\pi^0_l(c)=-\pi^0_s(-c)$. In particular, the dual representation of $\pi^0_l$ in the general becomes
\[
\pi^0_l(c) = \inf_{q\in\Q^a}\left\{\langle c,q\rangle + \sigma_\A(q)\,|\,q_0=1\right\},
\]
while in the conical case
\begin{equation*}\label{gdbs}
\pi^0_l(c) = \inf\left\{\left.\langle c,q\rangle \,\right|\,q\in\C^*\cap\B^*,\ q_0=1\right\}.
\end{equation*}

If we ignore the financial market (by setting $S\equiv 0$), the last part of Theorem~\ref{thm:acc} gives a multiperiod extension of \cite[Theorem~4.115]{fs6} to extended real-valued random loss functions. Besides the generality, the convex analytic proof above is considerably simpler.

\subsection{Indifference swap rates}

This section gives a dual representation for the indifference swap rate $\pi_s(\bar c,p;c)$. The arguments involved are very similar to those in the previous section once we notice that the indifference swap rate can be expressed as
\[
\pi_s(\bar c,p;c) = \inf\{\alpha\,|\, c-\alpha p\in\A(\bar c)\},
\]
where $\A(\bar c)=\{c\in\M^a\,|\,\varphi(\bar c+c)\le\varphi(\bar c)\}$ is the set of claims that an agent with current financial position $\bar c\in\M^a$ deems acceptable given his risk preferences and ability to trade in the illiquid market described by $S$ and $D$.

Analogously to the dual representation for risk measures and accounting values in the previous section, the dual representation for $\pi_s(\bar c,p;\cdot)$ below involves the support function of $\A(\bar c)$. Under mild conditions, the support function splits again into two terms. One of them is still the support function of the set $\C$ of claims that can be superhedged without a cost but now the second term is the support function of the set
\[
\B(\bar c) = \{c\in\M^a\,|\,EV(\bar c+c)\le\varphi(\bar c)\}.
\]
This is the set of claims an agent with financial position $\bar c\in\M^a$ and risk preferences $EV$ would deem at least as desirable as the possibility to trade in a financial market described by $S$ and $D$.

\begin{theorem}\label{thm:idsr}
If the assumptions of Theorem~\ref{thm:cl} are satisfied and $V^\infty\ge 0$, then, for every $\bar c\in\dom\varphi$ and $p\in\M^a$, the conditions
\begin{enumerate}
\item
$\pi_s(\bar c,p;\cdot)$ is proper and lower semicontinuous,
\item
$\pi_s(\bar c,p;0)>-\infty$,
\item
$p\notin\C^\infty$,
\item
$\langle p,q\rangle=1$ for some $q\in\dom\sigma_\C$
\end{enumerate}
are equivalent and imply the validity of the dual representation
\begin{align*}
\pi_s(\bar c,p;c) &= \sup_{q\in\Q^a}\left\{\left.\langle c,q\rangle - \sigma_{\A(\bar c)}(q)\,\right|\, \langle p,q\rangle=1\right\}
\end{align*}
If the above conditions hold for some $\bar c\in\dom\varphi$, then they hold for every $\bar c\in\dom\varphi$. If $\inf\varphi<\varphi(\bar c)$, then under the assumptions of Theorem~\ref{thm:varphi*}, $\sigma_{\A(\bar c)}=\sigma_{\B(\bar c)}+\sigma_\C$, where
\[
\sigma_{\B(\bar c)}(q) = \inf_{\alpha>0}\alpha[EV^*(q/\alpha)-\langle\bar c,q/\alpha\rangle-\varphi(\bar c)].
\]
\end{theorem}

\begin{proof}
The closedness of $\varphi$ in Theorem~\ref{thm:cl} implies the closedness of $\A(\bar c)$ and then, $\A^\infty(\bar c)=\{c\in\M^a\,|\,\varphi^\infty(c)\le 0\}$, by \cite[Corollary~6B]{roc66}. The first claim is now proved just like in Theorem~\ref{thm:acc}.

By Lagrangian duality, the condition $\inf\varphi<\varphi(\bar c)$ implies
\begin{align*}
\sigma_{\A(\bar c)}(q) &= \sup_{c\in \M^a}\{\langle c,q\rangle\,|\,\varphi(\bar c + c)-\varphi(\bar c)\le 0\}\\
&= \inf_{\alpha>0}\sup_{c\in \M^a}\{\langle c,q\rangle-\alpha[\varphi(\bar c + c)-\varphi(\bar c)]\}\\
&= \inf_{\alpha>0}\sup_{c\in \M^a}\{\langle c,q\rangle-\alpha[\varphi(c)-\varphi(\bar c)]\} - \langle\bar c,q\rangle\\
&= \inf_{\alpha>0}\alpha[\varphi^*(q/\alpha)-\varphi(\bar c)] - \langle\bar c,q\rangle\\
&= \inf_{\alpha>0}\alpha[EV^*(q/\alpha)-\varphi(\bar c)] - \langle\bar c,q\rangle + \sigma_\C(q),
\end{align*}
where the last equality comes from Corollary~\ref{cor:support2} and the fact that $\sigma_\C$ is positively homogeneous. By Lagrangian duality again
\begin{align*}
\sigma_{\B(\bar c)}(q) &= \sup_{c\in\M^a}\{\langle c,q\rangle\,|\,EV(\bar c+c)\le\varphi(\bar c)\}\\
&= \inf_{\alpha>0}\alpha[EV^*(q/\alpha)-\varphi(\bar c)] - \langle\bar c,q\rangle,
\end{align*}
just like in the proof of Theorem~\ref{thm:acc}.
\end{proof}

The structure and assumptions of Theorem~\ref{thm:idsr} are essentially the same as those in Theorem~\ref{thm:acc}. The main difference is in the interpretations. While the accounting value looks for least amount of initial cash that allows one to find an acceptable hedging strategy, the indifference swap rate compares two financial positions. The difference is reflected in the definitions of the ``acceptance sets'' $\B$ and $\B(\bar c)$, where the latter compares claims with the current financial position of a rational agent who has access to financial markets.

Specializing to market models with perfectly liquid cash, we obtain illiquid discrete-time versions of pricing formulas obtained in e.g.\ \cite{oz9} and~\cite{bfg11}.

\begin{example}[Numeraire and martingale measures]
Let $p=(1,0,\ldots,0)$ and assume that 
\[
S_t(x,\omega) = x^0 + \tilde S_t(\tilde x,\omega)\quad\text{and}\quad D_t(\omega) = \reals\times\tilde D_t(\omega)
\]
with a sublinear $\tilde S$ and conical $\tilde D$ as in Example~\ref{ex:liquid}. Like at the end of Section~\ref{sec:ALM}, we can then write the dual representation in terms of probability measures as
\begin{align*}
\pi(\bar c;c) &= \sup_{y\in\C^*}\left\{\left.\langle\bar c+c,y\rangle - \sigma_{\A(\bar c)}(y)\,\right|\, \langle p,y\rangle=1\right\}\\
&= \sup_{Q\in\P}\sup_{\alpha>0}\left\{E^Q\sum_{t=0}^T(\bar c_t+c_t)-\alpha\left[E\sum_{t=0}^Tv_t^*(E_t\frac{dQ}{dP}/\alpha)-\varphi(\bar c)\right]\right\}.
\end{align*}
This can be seen as an illiquid discrete-time version of the pricing formulas in~\cite[Section~7]{oz9} and~\cite[Section~4]{bfg11}.
\end{example}


\section{Contingent claims with physical delivery}\label{sec:physical}

In this section, we study problems of the form
\begin{equation}\label{alm+}\tag{ALM+}
\minimize\quad E V(S(\Delta x+\theta)+c)\quad\ovr\quad x\in\N_D,
\end{equation}
where $\theta$ is a $\reals^J$-valued process that can be interpreted as a contingent claim with {\em physical delivery} (portfolio-valued contingent claims) that the agent has to deliver in addition to the cash-settled claim $c$. Superhedging of contingent claims with physical delivery has been studied e.g.\ in \cite{kab99,pp10,kre9}. Our formulation is close to \cite{gr15} who studied optimal investment and superhedging under superlinear trading costs in an unconstrained continuous time market model with perfectly liquid cash.

We will denote the optimum value of \eqref{alm+} by $\bar\varphi(\theta,c)$. 
Clearly, the optimum value function $\varphi$ of \eqref{alm} is simply the restriction of $\bar\varphi(0,\cdot)$ to $\M^a$. Combined with some functional analytic arguments, this simple identity allow us to derive the main results of the previous sections from corresponding results on $\bar\varphi$. The introduction of the extra parameter $\theta$, in fact, simplifies the analysis and provides extra information about \eqref{alm}.

\subsection{Nonadapted claims}

We start by analyzing \eqref{alm+} with possibly nonadapted claims $(\theta,c)$. More precisely, we will study the optimum value $\bar\varphi(\theta,c)$ of \eqref{alm+} on the space
\[
\U:=L^\infty(\Omega,\F,P;\reals^{J(T+1)})\times \M
\]
of claim processes $u=(\theta,c)$ whose physical component $\theta$ is essentially bounded and the cash component $c$ belongs to $\M$. In most situations, one is interested in adapted claims but our formulation allows also for situations where the investor may remain unaware of the exact claim amounts until a certain future time (as may happen e.g.\ in asset management departments of large financial institutions). The space $\U$ is in separating duality with
\[
\Y:=L^1(\Omega,\F,P;\reals^{J(T+1)})\times \Q
\]
under the bilinear form $\langle u,y\rangle:=E(u\cdot y)$. We will split the dual variables $y$ into two processes $w$ and $q$ corresponding to the splitting of $u$ into $\theta$ and $c$. One can thus express the bilinear form as
\[
\langle u,y\rangle = \langle\theta,w\rangle + \langle c,q\rangle.
\]

The following result shows that, under Assumption~\ref{ass1}, the conjugate of $\bar\varphi$ can be expressed in terms of the loss function and the market model. The proof is given in the appendix.

\begin{theorem}\label{thm:barvarphi*}
If $V$ satisfies Assumption~\ref{ass1}, then the conjugate of the optimum value function $\bar\varphi$ of \eqref{alm+} with respect to the pairing of $\U$ with $\Y$ can be expressed as
\[
\bar\varphi^*(y) = E\left\{V^*(q) + \sum_{t=0}^T[(q_tS_t)^*(w_t) + \sigma_{D_t}(E_t\Delta w_{t+1})]\right\}.
\]
\end{theorem}

The next result gives sufficient conditions for the closedness of $\bar\varphi$ as well as an expression for the recession function $\bar\varphi^\infty$. The proof can be found in the appendix.

\begin{theorem}\label{thm:barcl}
If Assumptions~\ref{ass1}, \ref{ass3} and \ref{scale} hold and there exists $y\in\Y$ such that
\[
E\left\{V^*(q) + \sum_{t=0}^T[(q_tS_t)^*(w_t) + \sigma_{D_t}(E_t\Delta w_{t+1})]\right\}<\infty,
\]
then $\bar\varphi$ is $\sigma(\U,\Y)$-closed in $\U$, the infimum in \eqref{alm+} is attained for all $u\in\U$ and 
\begin{align*}
\bar\varphi^\infty(u) &= \inf_{x\in\N}EV^\infty(S^\infty(\Delta x+\theta)+c).
\end{align*}
\end{theorem}

When $V=\delta_{\reals^{J+1}_-}$, we have $\bar\varphi=\delta_{\bar\C}$, where
\[
\bar\C:=\{(\theta,c)\,|\,\exists x\in\N_D:\ S(\Delta x+\theta)+c\le 0\ P\text{-a.s.}\}
\]
is the set of multivariate claim processes that can be superhedged without a cost. The closedness result above can then be seen as a discrete-time version of \cite[Proposition~3.5]{gr15} which addresses unconstrained continuous-time models with perfectly liquid cash. When $V=\delta_{\reals^{J+1}_-}$, the assumptions of Theorem~\ref{thm:barcl} reduce to the requirement that
\[
\{x\in\N_{D^\infty}\,|\, S^\infty(\Delta x)\le 0\}
\]
is a linear space; see the discussion after Assumption~\ref{ass3} in Section~\ref{sec:cl}. This certainly holds under the superlinear growth condition of \cite{gr15}, which (in the notation of Example~\ref{ex:liquid}) implies that there is a strictly positive adapted process $H$ such that
\[
\tilde S_t(x,\omega)\ge H_t(\omega)|\tilde x|^\alpha\quad\forall\tilde x\in\reals^J.
\]
Indeed, this implies that $\tilde S_t^\infty=\delta_{\{0\}}$ so the linearity condition means that $\{(x^0,0)\in\N_D\,|\,\Delta x^0\le 0\}$ is a linear space. But this is obvious since $x_{-1}=0$ and $D_T=\{0\}$. Note that, unlike the superlinear growth condition, the linearity condition above does allow for cost functions $S$ which are decreasing in some directions, which is quite a natural assumption for assets with free disposal.

Under the assumptions of Theorems~\ref{thm:barvarphi*} and \ref{thm:barcl}, the optimum value of \eqref{alm+} equals that of the {\em dual} problem
\[
\minimize_{(w,q)\in\Y}\quad E\left\{V^*(q) + \sum_{t=0}^T[(q_tS_t)^*(w_t) + \sigma_{D_t}(E_t\Delta w_{t+1}) - \theta_t\cdot w_t - c_tq_t]\right\}.
\]
We know from the general conjugate duality theory that the dual solutions are then the subgradients of the optimum value function $\bar\varphi$ at $u=(\theta,c)$. Much like in \cite{roc74}, the dual variables can be used to give optimality conditions for the solutions of the primal problem \eqref{alm+}. The situation here is slightly different, however, since in \eqref{alm+}, the primal solutions are sought over the vector space $\N$ that lacks an appropriate locally convex topology. Applying the optimality conditions derived for convex stochastic optimization in \cite[Section~3]{bpp}, we obtain the following. The proof is given in the appendix.

\begin{theorem}\label{thm:baroc}
If $\bar\varphi$ is closed at $(\theta,c)$ and Assumption~\ref{ass1} holds, then $x\in\N$ solves \eqref{alm+} and $y\in\Y$ solves the dual problem if and only if they are feasible and
\begin{align*}
E_t\Delta w_{t+1} &\in N_{D_t}(x_t),\\
q &\in\partial V(S(\Delta x+\theta)+c),\\
w_t&\in\partial(q_tS_t)(\Delta x_t+\theta_t).
\end{align*}
Conversely, if $x$ and $y$ are feasible primal-dual pair satisfying the above system, then $x$ solves the primal, $(w,q)$ solves the dual, and $\varphi$ is closed at $(\theta,c)$.
\end{theorem}

\subsection{Adapted claims}\label{sec:adapted}

When the claim process $(\theta,c)$ is adapted and the loss function $V$ satisfies Assumption~\ref{ass2}, the above results can be expressed in terms of adapted dual variables. We will denote the linear subspaces of adapted processes in $\U$ and $\Y$ by $\U^a$ and $\Y^a$, respectively. In other words,
\[
\U^a:=\N^\infty\times\M^a\quad\text{and}\quad \Y^a:=\N^1\times\Q^a.
\]
Since, by assumption, $\M$ and $\Q$ are closed under adapted projections, we have that $\U^a$ and $\Y^a$ are in separating duality under the bilinear form defined for $\U$ and $\Y$.

We will denote the restriction of $\bar\varphi$ to $\U^a$ by $\bar\varphi_a$. Since the relative topology $\sigma(\U,\Y)$ on $\U^a$ coincides with $\sigma(\U^a,\Y^a)$, Theorem~\ref{thm:barcl} implies that $\bar\varphi_a$ is closed with respect to $\sigma(\U^a,\Y^a)$. It is immediate from the definition of the recession function that $\bar\varphi_a^\infty$ is simply the restriction of $\bar\varphi^\infty$ to $\U^a$. To restrict the dual variables to $\Y^a$, only one simple observation is needed.

If $\psi$ is any functional on $\U$ such that $\psi(\ap u)\le\psi(u)$ for all $u\in\U$, then for any $y\in\Y$,
\begin{align*}
\psi^*(\ap y)= \sup_{u\in\U}\{\langle u,\ap y\rangle - \psi(u)\} \le \sup_{u\in\U}\{\langle\ap u,y\rangle - \psi(\ap u)\} \le \psi^*(y).
\end{align*}
In particular, if $V$ satisfies Assumptions~\ref{ass1} and \ref{ass2}, then by the interchange rule, $EV^*(\ap q)\le V^*(q)$ so $V^*$ satisfies Assumption~\ref{ass2} as well. The following shows that Assumption~\ref{ass2} is inherited by $\bar\varphi^*$ and thus by the closure of $\bar\varphi$ as well.

\begin{theorem}\label{thm:barvarphia*}
If $V$ satisfies Assumptions~\ref{ass1} and \ref{ass2}, then $\bar\varphi^*(\ap y)\le\bar\varphi^*(y)$ for all $y\in\Y$ and $\bar\varphi_a^*=\bar\varphi^*$ on $\Y^a$. In particular, if $(\theta,c)\in\U^a$ and $y\in\Y$ solves the dual, then $\ap y\in\Y^a$ solves the dual as well and $x\in\N$ solves \eqref{alm+} if and only if it is feasible and satisfies the subdifferential conditions of Theorem~\ref{thm:baroc} with $\ap y$.
\end{theorem}

\begin{proof}
Since $S$ is adapted and $S(0)=0$, we have 
\[
E[(q_tS_t)^*(\ap w_t) + \sigma_{D_t}(E_t\Delta\ap w_{t+1})] \le E[(q_tS_t)^*(w_t) + \sigma_{D_t}(E_t\Delta w_{t+1})],
\]
by Jensen's inequality for normal integrands; see e.g.\ \cite[Corollary~2.1]{pen11c}. The first claim thus follows from the expression for $\bar\varphi^*$ in Theorem~\ref{thm:barvarphi*} since $EV^*(\ap q)\le EV^*(q)$ under Assumption~\ref{ass2}. For any $y\in\Y^a$, the conjugate of $\bar\varphi_a$ can be expressed as
\begin{align*}
\bar\varphi_a^*(y) &= \sup_{u\in\U^a}\{\langle u,y\rangle - \bar\varphi(u)\}\\
&= \sup_{u\in\U}\{\langle u,y\rangle - (\bar\varphi+\delta_{\U^a})(u)\}\\
&= \cl\inf_{y'\in\Y}\{\bar\varphi^*(y-y') + \delta_{(\U^a)^\perp}(y')\}\\
&= \bar\varphi^*(y),
\end{align*}
where the closure is taken with respect to the pairing of $\Y$ with $\U$ and the last equality holds by the first claim.
\end{proof}

The proofs of Theorems~\ref{thm:varphi*}, \ref{thm:cl} and \ref{thm:oc} are now simple applications of the Theorems~\ref{thm:barcl} and \ref{thm:barvarphia*} and the fact that $\varphi$ is the restriction of $\bar\varphi(0,\cdot)$ to the space $\M^a$ of adapted cash-valued claims.

\begin{proof}(of Theorem~\ref{thm:varphi*})
Defining $\gamma_q:\N^\infty\to\ereals$ by
\[
\gamma_q(\theta) := \inf_{c\in\M^a}\{\bar\varphi_a(\theta,c) - \langle c,q\rangle\},
\]
we have $\varphi^*(q)=-\gamma_q(0)$ and $\gamma_q^*(w)=\bar\varphi^*_a(w,q)$ so that 
\[
-\gamma_q^{**}(0)=\inf_{w\in\N^1}\bar\varphi^*_a(w,q),
\]
which, by Theorem~\ref{thm:barvarphia*}, is the desired expression. It thus suffices to show that $\gamma_q(0)=\gamma_q^{**}(0)$ and that the infimum is attained. By \cite[Theorem~17]{roc74}, this holds if $\gamma_q$ is bounded from above on a Mackey-neighborhood of the origin. Choosing $x=0$ and $c=-S(\theta)$, we get
\begin{align*}
\gamma_q(\theta) &= \inf_{c\in\M^a,x\in\N_D}E[V(S(\Delta x+\theta)+c) - c\cdot q]\le E\sum_{t=0}^Tq_tS_t(\theta_t).
\end{align*}
Convexity of $S(\cdot,\omega)$ and the assumption that $S(x,\cdot)\in\M$ for all $x\in\reals^J$ implies $S(x)\in\M$ for all $x\in L^\infty$, so \cite[Theorem~22]{roc74} implies that the last expression is Mackey-continuous on $L^\infty$ and thus bounded above on a neighborhood of the origin.
\end{proof}

\begin{proof}(of Theorem~\ref{thm:cl})
Theorem~\ref{thm:cl} follows from Theorem~\ref{thm:barcl} once we notice that $\varphi$ is the restriction of $\bar\varphi(0,\cdot)$ to $\M^a$ and that $\sigma(\M^a,\Q^a)$ is the relative topology of $\sigma(\U,\Y)$ on $\M^a$. 
\end{proof}

\begin{proof}(of Theorem~\ref{thm:oc})
It suffices to apply Theorem~\ref{thm:barvarphia*} with $\theta=0$.
\end{proof}

\section*{Appendix}

A set-valued mapping $\omega\mapsto C(\omega)$ is {\em $\F_t$-measurable} is the inverse image
\[
C^{-1}(U):=\{\omega\in\Omega\,|\,C(\omega)\cap U\ne\emptyset\}
\]
of every open $U\subset\reals^n$ is $\F_t$-measurable. An extended real-valued function $f$ on $\reals^n\times\Omega$ is a {\em normal integrand} if the set-valued mapping $\omega\mapsto\epi f(\cdot,\omega):=\{(x,\alpha)\in\reals^n\times\reals\,|\, f(x,\omega)\le\alpha\}$ is closed-valued and measurable. We refer the reader to \cite[Chapter~14]{rw98} for further details on measurable set-valued maps and normal integrands.

\subsection*{Compositions of convex functions}

Given an extended real-valued function $g$ on $\reals^m$ and an $\reals^m$-valued function $F$ on a subset $\dom F$ of $\reals^n$, we define their composition as the extended real-valued function
\[
(g\comp F)(x):=
\begin{cases}
  g(F(x)) & \text{if $x\in\dom F$},\\
  +\infty & \text{if $x\notin\dom F$}.
\end{cases}
\]
Given a convex cone $K\subset\reals^m$, the function $F$ is said to be {\em $K$-convex} if the set
\[
\epi_KF:=\{(x,u)\,|\,x\in\dom F,\ F(x)-u\in K\}
\]
is convex. It is said to be {\em closed} if $\epi_KF$ is a closed set. It is easily verified (see the proof of Lemma~\ref{lem:comp} below) that if $g$ is convex and $F$ is $K$ convex then $g\comp F$ is convex if
\begin{equation}\label{growthcomp}
F(x)-u\in K\implies g(F(x))\le g(u)\quad\forall x\in\dom F.
\end{equation}
For such composition, full subdifferential and recession calculus is available; see \cite{pen99}. The composition is well-behaved also in terms of measurability.

We say that a family $\{F(\cdot,\omega)\}_{\omega\in\Omega}$ of closed $K$-convex functions is a {\em random $K$-convex function} if $\epi_KF(\cdot,\omega)$ is measurable.

\begin{lemma}\label{lem:comp}
Let $g$ be a convex normal integrand and let $F$ be a random $K$-convex function such that \eqref{growthcomp} holds almost surely. If $(-K)\cap\{u\,|\,g^\infty(u,\omega)\le 0\}$ is linear, then $g\comp F$ is a convex normal integrand on $\reals^n\times\Omega$. 
\end{lemma}

\begin{proof}
By \cite[Example~14.32 and Proposition~14.44]{rw98},
\[
h(x,u,\omega) := g(u,\omega) + \delta_{\epi_KF(\cdot,\omega)}(x,u),
\]
is a normal integrand. The growth condition gives $(g\comp F)(x,\omega) =  \inf_{u\in\reals^m}h(x,u,\omega)$ while the linearity condition implies, by \cite[Theorem~9.2]{roc70a}, that this expression is lsc in $x$ and thus a normal integrand, by \cite[Proposition~14.47]{rw98}. 
\end{proof}

In the market model of Section~\ref{sec:oicv}, the function $S(x,\omega):=(S_t(x,\omega))_{t=0}^T$ with
\[
\dom S(\cdot,\omega):=\bigcap_{t=0,\ldots,T}\dom S_t(\cdot,\omega)
\]
defines a random $K$-convex function with $K:=\reals^{T+1}_-$. By \cite[Corollary~8.6.1]{roc70a}, Assumption~\ref{ass1} means that $(-K)\cap\{c\in\reals^{T+1}\,|\,V^\infty(c,\omega)\le 0\}=\{0\}$, so Lemma~\ref{lem:comp} gives the following.

\begin{corollary}\label{cor:vs}
Under Assumption~\ref{ass1}, the function $(x,c,\omega)\mapsto V(S(\Delta x,\omega))+c)$ is a convex normal integrand and, in particular, $\B(\reals^{J(T+1)}\times\reals^{(T+1)})\otimes\F$-measurable.
\end{corollary}

\subsection{Proofs of Theorems~\ref{thm:barvarphi*}, \ref{thm:barcl} and \ref{thm:baroc}}

As in \cite{pen14}, we analyze \eqref{alm+} in the more general parametric stochastic optimization format from \cite{pen11c,bpp}. To this end, we express the value function as
\[
\bar\varphi(u)= \inf_{x\in\N}\int f(x(\omega),u(\omega),\omega)dP(\omega),
\]
where $\N$ is the linear space of adapted $\reals^J$-valued processes and $f$ is the extended real-valued function on $\reals^{J(T+1)}\times\reals^{(J+1)(T+1)}\times\Omega$ defined by
\[
f(x,u,\omega)= V(S(\Delta x+\theta,\omega)+c,\omega)+\delta_{D(\omega)}(x).
\]
Indeed, $f$ is a normal integrand by, \cite[Proposition~14.44 and Proposition~14.45]{rw98} and Corollary~\ref{cor:vs}.

The dual expressions involve the associated {\em Lagrangian integrand}
\begin{align}
l(x,y,\omega)&=\inf_{u\in\reals^{(J+1)(T+1)}}\{f(x,u,\omega)-u\cdot y\}\nonumber\\
&=\delta_{D(\omega)}(x) - V^*(q,\omega) + \sum_{t=0}^T[\Delta x_t\cdot w_t - (q_tS_t)^*(w_t,\omega)]\nonumber\\
&=\delta_{D(\omega)}(x) - V^*(q,\omega) + \sum_{t=0}^T[- x_t\cdot\Delta w_{t+1} - (q_tS_t)^*(w_t,\omega)]\label{le}
\end{align}
and the conjugate of $f$
\begin{align*}
f^*(v,y,\omega)&=\sup_{\stackrel{x\in\reals^{J(T+1)}}{u\in\reals^{(J+1)(T+1)}}}\{x\cdot v+u\cdot y-f(x,u,\omega)\}\\
&=\sup_{x\in\reals^{J(T+1)}}\{x\cdot v-l(x,y,\omega)\}\\
&=  V^*(q,\omega) + \sum_{t=0}^T[\sigma_{D_t(\omega)}(v_t+\Delta w_{t+1}) + (q_tS_t)^*(w_t,\omega)].
\end{align*}
By \cite[Theorem~14.50]{rw98}, $f^*$ is a normal integrand as well.

We define an auxiliary value function
\[
\tilde\varphi(u):=\inf_{u\in\N^\infty}Ef(x,u),
\]
where $\N^\infty:=\N\cap L^\infty$, the essentially bounded adapted trading strategies. 

\begin{proof}(of Theorem~\ref{thm:barvarphi*})
Since $\{x\in\N^\infty\,|\,\exists u\in\U:\ Ef(x,u)<\infty\}=\N^\infty$, Theorem~2 of \cite{bpp} says that conjugate of $\tilde\varphi$ can be expressed as
\begin{align*}
\tilde\varphi^*(y) &= -\inf_{x\in\N^\infty}El(x,y)\\
&= -\inf_{x\in\N^\infty}\left\{E\sum_{t=0}^T[\delta_{D_t}(x_t) - x_t\cdot E_t[\Delta w_{t+1}]- (q_tS_t)^*(w_t)] - V^*(q)\right\}\\
&= E\left\{\sum_{t=0}^T[\sigma_{D_t}(E_t\Delta w_{t+1}) + (q_tS_t)^*(w_t)] + V^*(q)\right\}
\end{align*}
where the second equality comes from properties of conditional expectation and the last from the interchange rule for integration and minimization; see \cite[Theorem~14.60]{rw98}. We have $Ef^*(v,y)=\tilde\varphi^*(y)$ for any $y\in\Y$ and $v_t=E_t\Delta w_{t+1}-\Delta w_{t+1}$. Thus, by \cite[Theorem~2]{bpp}, $\bar\varphi^*=\tilde\varphi^*$. 
\end{proof}

\begin{proof}(of Theorem~\ref{thm:barcl})
By \cite[Theorem~5 and Lemma~6]{per16}, it suffices to show that $\L:=\{x\in\N\mid f^\infty(x,0)\le 0\}$ is a linear space and that there exists $(v,y)\in\N^\perp\times\Y$ with $\lambda(v,y)\in\dom Ef^*$ for two different $\lambda>0$.

By \cite[Theorem~9.3]{roc70a} and \cite[Theorem~7.3]{pen99},
\[
f^\infty(x,u,\omega) =V^\infty(S^\infty(\Delta x+\theta,\omega)+c,\omega) + \delta_{D^\infty(\omega)}(x),
\]
so $\L$ is a linear space under Assumption~\ref{ass3}. The assumption in the theorem means that $Ef^*(v,y)$ is finite for some $y$ and $v=E_t\Delta w_{t+1}-\Delta w_{t+1}$. By Assumption~\ref{ass3}, $EV^*(\lambda q)$ is finite for some nonnegative $\lambda\ne 1$, and since the other terms in 
\begin{align*}
Ef^*(v,y) &= E[V^*(q) + \sum_{t=0}^T[(q_tS_t)^*(w_t) + \sigma_{D_t}(v_t +\Delta w_{t+1})]
\end{align*}
are positively homogeneous, we get $\lambda (v,y)\in\dom Ef^*$.
\end{proof}

\begin{proof}(of Theorem~\ref{thm:baroc})
By Theorem~\ref{thm:barvarphi*}, $y$ solves the dual if and only if it maximizes $\langle u,y\rangle-\bar\varphi^*(y)$. Under the closedness condition, this is equivalent to $y\in\partial\bar\varphi(u)$. The assumptions of \cite[Theorem~8]{bpp} are then satisfied with $v_t:=E_t\Delta w_{t+1}-\Delta w_{t+1}$. Thus, $x\in\N$ and $y\in\Y$ are optimal if and only if they are feasible and 
\begin{align*}
v&\in\partial_xl(x,w,q),\\
(\theta,c) &\in\partial_{(w,q)} [-l](x,w,q)
\end{align*}
$P$-almost surely. Looking at the last expression \eqref{le} for $l$, we see that the first inclusion here coincides with the first one in the statement. The second inclusion means that $c\in\partial V^*(q)+d$ and $\Delta x+\theta =\eta$ for some $(\eta_t,d_t)\in\partial_{q,w}(q_tS_t)^*(w_t)$. Indeed, the subdifferential sum rule \cite[Theorem~23.8]{roc70a} applies $\omega$-wise since Assumption~\ref{ass1} implies that $\dom V^*\cap\reals^{T+1}_{++}\ne\emptyset$ while, by \cite[Theorem~6.6]{roc70a}, the relative interior of the domain of the function $(w,q)\mapsto(qS_t)^*(w,\omega)$ contains a $(w,q)$ whenever $q\in\reals^{T+1}_{++}$. It is easily checked that $(q_tS_t)^*(w_t)$ is the support function of the set $E_t:=\{(\eta,d)\,|\,S_t(\eta)+d\le 0\}$. By \cite[Theorem~23.5 and p. 215]{roc70a}, $(\eta_t,d_t)\in\partial_{q,w}(q_tS_t)^*(w_t)$ is thus equivalent to $(w_t,q_t)\in N_{E_t}(\eta_t,d_t)$. This means that either $S_t(\eta_t)+d_t<0$ and $(w_t,q_t)=(0,0)$ or $S_t(\eta_t)+d_t=0$ and $w_t\in\partial(q_tS_t)(\eta_t)$. The second inclusion thus means that
\begin{align*}
q &\in\partial V(c-d),\\
w_t&\in\partial(q_tS_t)(\Delta x_t+\theta_t)
\end{align*}
for some measurable process $d\le-S(\Delta x+\theta)$ such that $q(S(\Delta x+\theta)+d)=0$ almost surely. Since $V$ is nondecreasing, this is satisfied also by $d=-S(\Delta x+\theta)$.
\end{proof}

\subsection*{A convex analytic lemma}

The following was used in the proofs of Theorems~\ref{thm:acc} and \ref{thm:idsr}.

\begin{lemma}\label{lem:piD}
Let $\D\subset\M^a$ be a closed set containing the origin, let $p\in\M^a$ and 
\[
\pi(c)=\inf\{\alpha\in\reals\,|\,c-\alpha p\in\D\}.
\]
Then the conditions
\begin{enumerate}
\item
$\pi(c)$ is closed and proper,
\item
$\pi(0)>\infty$,
\item
$p\notin\D^\infty$,
\item
$\langle p,q\rangle=1$ for some $q\in\dom\sigma_\D$
\end{enumerate}
are equivalent and imply the validity of the dual representation
\[
\pi^*(c) = \sup_{q\in\Q^a}\{\langle c,q\rangle - \sigma_\D(q)\,|\, \langle p,q\rangle =1\}.
\]
\end{lemma}

\begin{proof}
This is proved for $\D=\C$ in \cite[Proposition~4.2 and Theorem~5.2]{pen11b}. The same arguments work for any convex $\D$.
\end{proof}

\subsection*{Asymptotic elasticity of a random loss function}

The main result of \cite{pp12} was used in \cite{pen14} to give condition for the lower semicontinuity of the optimal value function of \eqref{alm}. One of the conditions was that the functions $V_t$ are bounded from below. While this is satisfied in many cases, it rules out isoelastic utility functions. Theorem~5 in \cite{per16} allows us to replace the lower bound by more general ``asymptotic elasticity'' conditions introduced in \cite{ks99} and \cite{sch1}. We extend the conditions of \cite{ks99,sch1} further by allowing for nonsmooth and random functions $V_t$ and by not imposing the Inada conditions, which can be written as $V_t^\infty(\cdot,\omega)=\delta_{\reals_-}$. Indeed, it suffices to have one of the conditions in the following lemma. 

\begin{lemma}\label{lem:scale}
Let $h$ be a convex normal integrand on the real line. If there exist $\lambda>1$, $\bar y\in\dom Eh^*$ and $C\ge 0$ such that
\[
h^*(\lambda y,\omega)\le Ch^*(y,\omega)\quad\forall y\ge \bar y(\omega)
\]
then (in $L^0$)
\[
\lambda\dom Eh^*\subseteq\dom Eh^*\quad\forall\lambda\ge 1.
\]
If there exist $\lambda\in(0,1)$, $\bar y\in\dom Eh^*$ and $C>0$ such that
\[
h^*(\lambda y,\omega)\le Ch^*(y,\omega)\quad\forall y\in[0,\bar y(\omega)]
\]
then
\[
\lambda\dom Eh^*\subseteq\dom Eh^*\quad\forall\lambda\in(0,1].
\]
\end{lemma}

\begin{proof}
Let $y\in\dom Eh^*$. Under the first condition
\begin{align*}
h^*(\lambda y) &= \one_{\{y\le \bar y\}}h^*(\lambda y) + \one_{\{y> \bar y\}}h^*(\lambda y)\\
&\le \one_{\{y\le \bar y\}}\max\{h^*(y),h^*(\lambda \bar y)\} + \one_{\{y> \bar y\}}h^*(\lambda y)\\
&\le \one_{\{y\le \bar y\}}\max\{h^*(y),Ch^*(\bar y)\} + \one_{\{y> \bar y\}}Ch^*(y),
\end{align*}
where the first inequality comes from the convexity of $h^*$. Since $\bar y,y\in\dom Eh^*$, the last expression is integrable.

Under the second condition,
\begin{align*}
Eh^*(\lambda y) &= E\one_{\{y\le \bar y\}}h^*(\lambda y) + E\one_{\{y> \bar y\}}h^*(\lambda y)\\
&\le E\one_{\{y\le \bar y\}}h^*(\lambda y) + E\one_{\{y> \bar y\}}\max\{h^*(y),h^*(\lambda \bar y)\}\\
&\le E\one_{\{y\le \bar y\}}Ch^*(y) + E\one_{\{y> \bar y\}}\max\{h^*(y),Ch^*(\bar y)\},
\end{align*}
where the first inequality comes from the convexity of $h^*$.
\end{proof}




The following lemmas give alternative formulations of the conditions in Lemma~\ref{lem:scale} which were used in the proof of Theorem~\ref{thm:barcl}. In both lemmas, third condition extends the ``asymptotic elasticity'' conditions introduced in \cite{ks99} and \cite{sch1}, respectively, for deterministic utility functions.

\begin{lemma}\label{lem:rae01}
Given a nondecreasing closed convex function $g$ on the real line with $g(0)=0$, the following are equivalent conditions on $\beta>1$ and $\bar y$ for which $\partial g^*(\bar y)\subset\reals_+$.
\begin{alignat}{3}
g^*(\lambda y)&\le \lambda^{\frac{\beta}{\beta-1}}g^*(y)\quad &&\forall \lambda\ge 1,\ y\ge \bar y.\label{1a}\\
cy &\le \frac{\beta}{\beta-1} g^*(y)\quad &&\forall y \ge \bar y,\ c\in\partial g^*(y).\label{1b}\\
cy &\ge \beta g(c)\quad &&\forall c\ge\partial g^*(\bar y),\ y\in\partial g(c).\label{1c}\\
g(\lambda c) &\ge \lambda^\beta g(c)\quad &&\forall \lambda\ge 1,\ c\ge\partial g^*(\bar y).\label{1d}
\end{alignat}
\end{lemma}
\begin{proof}
Assuming \eqref{1c}, we have $g(\lambda c)\ge g(c)+\int_{[c,\lambda c]}\beta\frac{g(s)}{s}ds$, so
\begin{align*}
g(\lambda c) &\ge g(c) +\int_{[c,\lambda c]}g(c)\frac{\beta}{s}\exp(\int_{[s,\lambda c]}\frac{\beta}{r}dr)ds = \lambda^\beta g(c)
\end{align*}
by  Gronwall's inequality. For $y\ge\bar y(\omega)$, condition \eqref{1d} gives
\begin{align*}
g^*(y) \ge \sup\{cy-\lambda^{-\beta}g(\lambda c)\}= \lambda^{-\beta}g^*(\lambda^{\beta-1}y),
\end{align*}
which implies \eqref{1a}. We have for every $\lambda> 1$, $y\ge \bar y(\omega)$, and $c\in\partial g^*(y)$,
\[
c(\lambda y-y)\le g^*(\lambda y)-g^*(y) \le (\lambda^{\frac{\beta}{\beta-1}}-1)g^*(y)
\]
from where we get \eqref{1b} by first dividing by $\lambda-1$ and then letting $\lambda\searrow 1$. Using the equivalence
\[
c\in\partial g^*(y,\omega)\iff y\in\partial g(c)\iff g(c)+g^*(y)=cy
\]
we can write condition \eqref{1b} as 
\[
cy\le\frac{\beta}{\beta-1}(cy-g(c))\quad\forall c\ge\partial g^*(\bar y),\ y\in\partial g(c)
\]
which is equivalent with \eqref{1d}.
\end{proof}

\begin{lemma}\label{lem:rae02}
Given a nondecreasing closed convex function $g$ on the real line with $g(0)=0$, the following are equivalent conditions on $\beta\in(0,1)$ and $\bar y$ for which $\partial g^*(\bar y)\subseteq\reals_-$.
\begin{alignat}{3}
g^*(\lambda y) &\le \lambda^{\frac{-\beta}{1-\beta}} g^*(y)\quad &&\forall \lambda\in(0,1),\ y\in(0,\bar y].\label{2a}\\
cy &\ge\frac{-\beta}{1-\beta} g^*(y)\quad &&\forall y\in(0,\bar y],\ c\in\partial g^*(y).\label{2b}\\
cy &\ge\beta g(c)\quad &&\forall c\le\partial g^*(\bar y),\ y\in\partial g(c).\label{2c}\\
g(\lambda c) &\ge \lambda^\beta g(c)\quad &&\forall \lambda\ge 1,\ c\le\partial g^*(\bar y).\label{2d}
\end{alignat}
\end{lemma}

\begin{proof}
Assuming \eqref{2c}, we have $g(c)\le g(\lambda c)+\int_{[\lambda c,c]}\beta\frac{g(s)}{s}ds$, so
\begin{align*}
g(c) &\le g(\lambda c) +g(\lambda c)\int_{[\lambda c,c]}\frac{\beta}{s}ds =g(\lambda c)(1-\ln(\lambda^{-\beta})) \le \lambda^{-\beta} g(\lambda c)
\end{align*}
For $y\in(0,\bar y]$, condition \eqref{2d} gives
\begin{align*}
g^*(y) \ge \sup\{cy-\lambda^{-\beta}g(\lambda c)\}= \lambda^{-\beta}g^*(\lambda^{\beta-1}y),
\end{align*}
which implies \eqref{2a}. We have for every $\lambda\in(0,1)$, $y\in(0,\bar y]$, and $c\in\partial g^*(y)$,
\[
c(\lambda y-y)\le g^*(\lambda y)-g^*(y) \le (\lambda^{\frac{-\beta}{1-\beta}}-1)g^*(y)
\]
from where we get \eqref{2b} by first dividing by $\lambda-1$ and then letting $\lambda\nearrow 1$. Using the equivalence
\[
c\in\partial g^*(y)\iff y\in\partial g(c)\iff g(c)+g^*(y)=cy
\]
we can write condition \eqref{2b} as 
\[
cy\ge\frac{-\beta}{1-\beta}(cy-g(c))\quad\forall c\ge\partial g^*(\bar y),\ y\in\partial g(c)
\]
which is equivalent with \eqref{2c}.
\end{proof}

\bibliographystyle{plain}
\bibliography{sp}

\end{document}